\documentclass[12pt]{article}
\usepackage{graphicx}
\usepackage[centering, margin={0.8in, 0.8in}, includeheadfoot]{geometry}

\usepackage[utf8x]{inputenc}
\usepackage{amssymb,pdfsync,epstopdf,verbatim,gensymb,amsmath,amsthm}
\usepackage{nccmath, mathtools,bm}
\usepackage{textcomp}
\usepackage{amsmath}
\usepackage{float}
\usepackage{subfig,boldline}
\usepackage{amsfonts}
\usepackage{epsfig,color, hyperref}
\usepackage{tabularx,longtable,adjustbox,siunitx,float,multirow,graphicx,booktabs}

\def\R {{\mathbb R}}

\def\cE {{\mathcal{E}}}

\def\H01{{H_0^1(\Omega)}}
\def\L2{{L^2(\Omega)}}

\newcommand*{\eps}[0]{{\epsilon}}

\usepackage{bmpsize}
\newtheorem{theorem}{Theorem}[section] 

\newtheorem{lemma}{Lemma}[section]

\newtheorem{remark}{Remark}[section]

\newtheorem{proposition}{Proposition}
\newtheorem{algorithm}{Algorithm}[section]
\DeclareMathOperator*{\argmin}{arg\,min}

\newcommand{\bT}{\boldsymbol{\theta}}
\newcommand{\bP}{\boldsymbol{P}}
\newcommand{\bps}{\boldsymbol{\psi}}
\newcommand{\bU}{\boldsymbol{U}}
\newcommand{\bX}{\boldsymbol{X}}
\newcommand{\bx}{\boldsymbol{x}}
\newcommand{\bF}{\boldsymbol{F}}
\newcommand{\bW}{\boldsymbol{W}}
\newcommand{\bS}{\boldsymbol{\sigma}}

\newcommand{\Beq}{\begin{equation}}
\newcommand{\Eeq}{\end{equation}}
\newcommand{\beq}{\begin{equation*}}
\newcommand{\eeq}{\end{equation*}}
\newcommand{\bal}{\begin{align}}
\newcommand{\eal}{\end{align}}

\renewcommand{\L}{\langle}

\newcommand{\bp}{\begin{prob}}
\newcommand{\ep}{\end{prob}}
\newcommand{\bpr}{\begin{proof}}
\newcommand{\epr}{\end{proof}}

\newcommand{\bel}[1]{\begin{equation}\label{#1}}
\newcommand{\ee}{\end{equation}}

\graphicspath{{./figures/}}

\date{}
\author{
Souvik Roy{\footnote{(Corresponding Author)
souvik.roy@uta.edu, Department of Mathematics, The University of Texas at Arlington, Arlington, TX 76019-0407, USA}}\and
Zui Pan{\footnote{zui.pan@uta.edu, College of Nursing and Health Innovation, The University of Texas at Arlington, Arlington, TX 76019-0407, USA}}\and
Suvra Pal{\footnote{suvra.pal@uta.edu, Department of Mathematics, The University of Texas at Arlington, Arlington, TX 76019-0407, USA}}
  }

\begin{document}

\title{A Fokker-Planck feedback control framework for optimal personalized therapies in colon cancer-induced angiogenesis }

\maketitle
\begin{abstract}
In this paper, a new framework for obtaining personalized optimal treatment strategies in colon cancer-induced angiogenesis is presented. The dynamics of colon cancer is given by a It\'o stochastic process, which helps in modeling the randomness present in the system. The stochastic dynamics is then represented by the Fokker-Planck (FP) partial differential equation (PDE) that governs the evolution of the associated probability density function. The optimal therapies are obtained using a three step procedure. First, a finite dimensional FP-constrained optimization problem is formulated that takes input individual noisy patient data, and is solved to obtain the unknown parameters corresponding to the individual tumor characteristics. Next, a sensitivity analysis of the optimal parameter set is used to determine the parameters to be controlled, thus, helping in assessing the types of treatment therapies. Finally, a feedback FP control problem is solved to determine the optimal combination therapies. Numerical results with the combination drug, comprising of Bevacizumab and Capecitabine, demonstrate the efficiency of the proposed framework.
\end{abstract}



\section{Introduction}

Colon cancer is a primary cause contributing to the worldwide cancer related deaths \cite{Fitz2013}. Due to the absence of symptoms at an early stage of colon cancer, its detection in patients usually occurs when the cancer becomes metastatic, due to the lack of early symptoms \cite{Schm2012}. Thus, it is of paramount importance to devise rapid and effective treatment strategies. For this purpose, one needs to have a proper understanding of the dynamics of relevant biomarkers to track the cancer progression and, further, determine the controllable tumor-sensitive biological factors. In this context, angiogenesis is an important biomarker that describes the formation of blood vessels, where new vasculature develops in order to support the tumor as it increases in size. Angiogenesis has been found to be a crucial prognostic factor in colon cancer \cite{Folk1,Folk2}. In the past, there have been numerous experimental studies related to the investigation of the angiogenesis induced by colon cancer in order to develop targeted therapies \cite{Baud2015}. For e.g., the authors in \cite{Ferr2003} study the dynamics of vascular endothelial growth factor (VEGF), which belongs to the class of tumor angiogenic factors (TAF) that promotes angiogenesis. In \cite{Ferr2005}, the authors investigate the role of a humanized monoclonal antibody targeting VEGF, known as Bevacizumab, as an anti-angiogenic drug. In \cite{Sapi2015}, the authors consider two different treatment strategies with low and high doses of Bevacizumab in experimental mice, and track the effects on tumor growth. In addition to Bevacizumab, effects of another anti-angiogenic drug called ziv-aflibercept was studied in \cite{Sun2012}. Recently, the authors in \cite{Yin2020} analyzed the role of HIF-$1\alpha$, VEGF and microvascular density in primary and metastatic tumors. For a comprehensive discussion on the angiogenesis pathway and some important biomarker discoveries for colon cancer, we refer the readers to \cite{Mousa2015}.

For treatment of colon cancer, it is important to devise optimal therapies that takes into account the maximum allowable drug amount and side-effects of the drugs. This process requires the need of experimental studies of drug effects with a combination of different drugs, like anti-angiogenic and chemotherapeutic drugs. However, since these drugs are quite costly, it is expensive to perform in-vitro and in-vivo experimental studies for testing effects of such drugs. This motivates the need to use computational frameworks as an alternate cost effective option for testing optimal therapies. For this purpose, it is important to understand the angiogenesis pathway mechanisms through mathematical pharmacokinetic models. Traditionally, pharmacokinetic models are described by a complex, non-linear system of deterministic differential equations, governing the dynamics of a biological process and drug interactions. There are several deterministic pharmacokinetic dynamical models that describe the process of angiogenesis in colon cancer. The authors in \cite{Bald1985} describe a diffusion-based model for TAF and growth of capillaries. In \cite{Chap1}, the authors describe the diffusion  of  the  TAF  into  the  surrounding  host  tissue  and  the  response  of  the endothelial  cells  to  the  chemotactic  stimulus through a system of partial differential equations (PDE). Another PDE-based model was used \cite{Chap2} to describe the features of a growing capillary network. In \cite{Byrne1995}, a PDE system is used to describe the migration of capillary sprouts in response to a chemoattractant field set up by TAF. Furthermore, a successful or failed neovascularization of the tumor is described through  the existence of traveling wave solutions of this system. The authors in \cite{Vila2017} develop a PDE-based model for vascular regression and regrowth. Since, the dynamics for angiogenesis contain inherent randomness due to phenotypic heterogenity \cite{Dick2008}, a more realistic dynamical model for angiogenesis can be given using stochastic differential equations. In this context, a few stochastic models have been developed to describe the process of angiogenesis. In \cite{Moha2008}, the authors develop a Markov chain model to investigate the changes of microvessel densities in tumor. A simplified stochastic geometric model was built by the authors in \cite{Capa2009} to describe a spatially distributed angiogenic processes. For a detailed review of the other available dynamical models for angiogenesis, we refer the readers to the review papers \cite{Chap3,Pierce2008,Rie2015}.

In addition, there are several works that use pharmacokinetic models to determine optimal treatments for angiogenesis-induced cancers.
In \cite{Chap4}, it was demonstrated that  administering anti-angiogenesis treatment first allows for more effective delivery of chemotherapy. In \cite{Powa2013}, the authors consider a pharmacokinetic cellular automata model to incorporate the cytotoxic effects of chemotherapy drugs. The authors in \cite{Argy2016} use a  two-compartmental model to capture pharmacokinetic properties of Bevacizumab into an ordinary differential equations (ODE)-based vascular tumour growth model. Recently, in \cite{Stur2018}, the authors determine an optimal treatment strategy for colon cancer-induced angiogenesis using a combination of Bevacizumab and FOLFOX (a chemotherapy drug). A major drawback was pointed out in \cite{Argy2016,Stur2018} to allude to the fact that the parameter estimation process was done using assembled data from disparate sources, such as multiple biological studies and clinical assays. Now, the coefficients in pharmacokinetic models, representing unknown parameters, describe an individual patient tumor characteristics. Since, the properties of tumor vary from patient to patient, an accurate estimation of these parameters is important in, subsequently, developing effective treatment strategies. Thus, the traditional parameter estimation methods results in an inadequate validation and are not useful in devising personalized therapies. Another drawback in the aforementioned works is that optimal treatment strategies are devised based on ODE-based pharmacokinetic models, that fail to capture the randomness in the dynamical process.

We contribute to the field of pharmacokinetic cancer research by presenting an effective approach to develop personalized therapies for colon cancer-induced angiogenesis. Our approach is based on a new coupled parameter estimation-sensitivity analysis technique, in the realm of PDE-constrained optimal control framework. The starting point of this estimation process is to consider a recent dynamical model for angiogenesis, given in \cite{Cser2019}. The model describes the evolution of three variables: the proliferating tumor volume, the vasculature volume in tumor and the dynamics of tumor angiogenic factors (TAF). To incorporate randomness of the tumor-induced angiogenesis dynamics, we extend the dynamical model presented in \cite{Cser2019}, to a It\^o stochastic process. To develop personalized therapies the first step is to determine the unknown coefficients or parameters of this stochastic process, that represent the individual-specific properties, from given patient data, by solving a PDE-constrained optimization problem. But due to the presence of random variables, one needs to consider expectation cost functionals for the optimization problem. To solve this problem, one can then use the method of dynamic programming to determine the necessary Hamilton-Jacobi-Bellman (HJB) equations. But this poses severe challenges due to the complexity nature of the underlying dynamical stochastic process. 

A more convenient framework for determining the unknown parameters is to use a deterministic setup through the Fokker-Planck (FP) equations, that represents the evolution of the joint probability density function (PDF) associated to the random variables in stochastic process. Usually, the experimental data contains random noise that arises due to the inherent cell measurement errors by different methods. Thus, while developing parameter estimation methods, one needs to incorporate the presence of noisy data into the estimation process. In this context, the FP optimization framework provides a robust mechanism to encompass a wide range of objective functionals that can incorporate noisy data measurements while providing accurate estimates of the parameters. Furthermore, one can use the FP open-loop and feedback control frameworks to efficiently solve highly non-linear optimal control problems. Such FP control frameworks have been used in past for problems arising in control of collective and crowd motion \cite{Roy2016,Roy2018a}, investigating pedestrian motion from a game theoretic perspective \cite{Roy2017}, reconstructing cell-membrane potentials\cite{Annun2021}, mean field control problems \cite{Borzi2020}, controlling production of Subtilin \cite{Thal2016}. To the best of our knowledge, this is the first work that considers the FP framework to devise optimal treatment strategies for cancer.

To determine the optimal combination therapies using a FP framework, we use a three step process: first, we formulate and solve a finite dimensional FP-constrained optimization problem to obtain the unknown patient-specific parameters. In the second step, we use a Monte Carlo-based sampling technique, called latin hypercube sampling (LHS), together with partial rank correlation coefficient (PRCC) analysis to determine the most sensitive parameters. In the third and final step, we formulate a feedback FP control problem, where the control functions represent feedback drug concentrations, corresponding to combination therapies. We use a combination therapy comprising of an anti-angiogenic inhibitor, Bevacizumab, and chemotherapeutic drug, Capecitabine, for our study. The FP control problem is solved to obtain the optimal combination drug dosages and optimal time of administering the drug. We remark that the motivation for considering feedback dosage concentrations is because, in clinical practice, effective dosages are administered not at regular intervals but in adaptive mode based on the parameter readings of the patient (for e.g., see the review article \cite{Almeida2018}). Feedback strategies for cancer treatment have been mathematically very less explored in literature \cite{Algoul2010,Tang2016}. But the authors in \cite{Tang2016} do suggest that feedback type treatments result in a better and optimal tumor control and drug dosage strategies. Also, in \cite{Tang2016}, it is mentioned that personalized therapies are more natural and feasible. The novelty of our work is the development of a new feedback treatment framework for obtaining personalized and rapid treatment mechanism for colon cancer-induced angiogenesis.

In the next section, we describe a three-step FP control framework for devising optimal treatments, which is based on solving two FP-optimization problems. Section \ref{sec:theory} is concerned with the theoretical properties of the FP optimization problems. In Section \ref{sec:numerical_FP}, we present a numerical scheme based on a combined splitting technique for time discretization and Chang-Cooper scheme for spatial discretization for the FP equations. We prove the properties of conservativeness, positivity, and second order convergence of the scheme. Section \ref{sec:uncertainty} is devoted to the theory of the uncertainty quantification and sensitivity analysis of the estimated parameter set using the LHS-PRCC technique. In Section \ref{sec:results}, we validate our framework by estimating parameters using synthetically generated data and real data from \cite{Cser2019,Sapi2015}. Furthermore, we apply the LHS-PRCC technique on both generated data and real data to identify the most sensitive parameters with respect to an output of our interest. Finally, we obtain the optimal combination therapies, comprising of Bevacizumab and Capecitabine, and show the correspondence of the results with experimental findings. A section of conclusion completes the exposition of our work.

\section{A Fokker-Planck framework for effective combination therapies}\label{sec:FP}

In this section, we present a Fokker-Planck framework for treatment assessment in colon cancer-induced angiogenesis. The starting point is to describe the dynamics of angiogenesis using a coupled system of ordinary differential equations (ODEs), based on the model given in \cite{Cser2019}. The following are the variables associated with different types of cell populations whose evolution we track over time.
\begin{enumerate}

\item $V(t)$-  the total tumor cell volume (cm$^3$)
\item $B(t)$-  the vasculature volume in the tumor (cm$^3$)
\item $T(t)$-  the concentration of tumor angiogenic factors (TAF) in the tumor (mg/ml),

\end{enumerate}
where $t$ is the time variable. The governing system of ODEs, representing the dynamics of the aforementioned variables, are given as follows
\begin{equation}\label{eq:ODE}
\begin{aligned}
&\dfrac{dV}{d\tau} =c\gamma V - \alpha_1 u_1(V,t)V,~ V(0) = V_0,\\
&\dfrac{dB}{d\tau} = c_{e}c\gamma V + c_v TB- \alpha_2 u_2(B,t)B,~ B(0) = B_0,\\
&\dfrac{dT}{d\tau} = c_T(1-\gamma) - q_TT- \alpha_3 u_3(T,t)T,~ T(0) = T_0.\\
\end{aligned}
\end{equation}
The unknown patient parameters that need to be determined is the parameter vector $\boldsymbol{\theta} = (c,c_e,c_v,c_T,q_T,\gamma)$, defined as follows
\begin{enumerate}
\item $c$- growth rate of tumor (day$^{-1}$)
\item $c_e$- rate of internalization of new vasculature from the environment
\item $c_v$- rate of formation of new blood vessels due to TAF (ml mg$^{-1}$ day$^{-1}$)
\item $c_T$- rate of production of TAF (mg ml$^{-1}$ day$^{-1}$)
\item $q_T$- rate of removal of TAF from tumor (day$^{-1}$)
\item $\gamma$- ratio of well-supported tumor cells inside the tumor volume
\end{enumerate}
The parameter $\gamma$ is one of the important parameters of interest as it determines the ratio of the tumor cells inside the tumor volume that receive nutrients from outside the tumor. In \cite{Cser2019}, $\gamma$ is a function of $t$ that initially decreases when the tumor volume is close to zero but stabilizes quickly after the volume reaches a threshold level. We consider our modeling framework with a non-zero starting tumor volume and, thus, we assume that $\gamma$ is constant.

The functions $u_1(V,t),u_2(B,t),u_3(T,t)$ represent dosages of combination of different chemotherapeutic and anti-angiogenic drugs, like Bevacizumab and Capecitabine. We consider the feedback dosages to be dependent only on the specific variable on which the drug affects the most \cite{Jack2000}. The constants $\alpha_i,~ i = 1,2,3$ represent efficiency of the drugs. For stabilization and scalability of the numerical algorithms, we non-dimensionalize the ODE system \eqref{eq:ODE} using the following non-dimensionalized state and time variables, and parameters
\begin{equation}\label{eq:Nvar}
\begin{aligned}
& \bar{V} = k_1 V,~ \bar{B} = k_2 B,~ \bar{T} = k_3 T,~t = k_4 \tau, \\
& \bar{c} = \frac{c}{k_4},~ \bar{c}_e = \frac{c_ek_2}{k_1},~ \bar{c}_v = \frac{c_v}{k_3k_4},~ \bar{c}_T = \frac{c_Tk_3}{k_4},~ \bar{q}_T = \frac{q_T}{k_4},~ \bar{u}_i = \frac{\alpha_i u_i}{k_4},~ i = 1,2,3.
\end{aligned}
\end{equation}
Then, the transformed non-dimensionless ODE system is given as follows
\begin{equation}\label{eq:NODE}
\begin{aligned}
&\dfrac{d\bar{V}}{dt} =\bar{c}\gamma \bar{V}- \bar{\alpha}_1 \bar{u}_1(\bar{V},t)\bar{V},~ \bar{V}(0) = \bar{V}_0,\\
&\dfrac{d\bar{B}}{dt} = \bar{c}_{e}\bar{c}\gamma \bar{V} + \bar{c}_v \bar{T}\bar{B} -\bar{\alpha}_2 \bar{u}_2(\bar{B},t)\bar{B},~ \bar{B}(0) = \bar{B}_0,\\
&\dfrac{d\bar{T}}{dt} = \bar{c}_T(1-\gamma) - \bar{q}_T\bar{T} - \bar{\alpha}_3 \bar{u}_3(\bar{T},t)\bar{T},~ \bar{T}(0) = \bar{T}_0.\\
\end{aligned}
\end{equation}

The system of ODEs given in \eqref{eq:NODE} can be written in a compact form as follows
\begin{equation}\label{eq:ODE_compact}
\begin{aligned}
&\dfrac{d\bX}{dt} = \bF(\bX,\bT, \bU(\bX,t)),\\
&\bX(0) = \bX_0,
\end{aligned}
\end{equation}
where $\bX(t) = (\bar{V}(t),\bar{B}(t),\bar{T}(t))^T$, and, without loss of generality, $\bT = (\bar{c},\bar{c}_e,\bar{c}_v,\bar{c}_T,\bar{q}_T,\gamma)$ and $\bU = (\bar{u}_1(\bar{V},t),\bar{u}_2(\bar{B},t),\bar{u}_3(\bar{T},t))$.

We extend the ODE system \eqref{eq:ODE} to the following system of It\^o stochastic differential equation corresponding to \eqref{eq:NODE} 
\begin{equation}\label{eq:ItoODE}
\begin{aligned}
&\dfrac{d\bar{V}}{dt} =\bar{c}\gamma \bar{V}-  \bar{u}_1(\bar{V},t)\bar{V}+ \sigma_1(\bar{V},\bar{B},\bar{T}) ~dW_1(t),~ \bar{V}(0) = \bar{V}_0,\\
&\dfrac{d\bar{B}}{dt} = \bar{c}_{e}\bar{c}\gamma \bar{V} + \bar{c}_v \bar{T}\bar{B}-\bar{u}_2(\bar{B},t)\bar{B}+ \sigma_2(\bar{V},\bar{B},\bar{T}) ~dW_2(t),~ \bar{B}(0) = \bar{B}_0,\\
&\dfrac{d\bar{T}}{dt} = \bar{c}_T(1-\gamma) - \bar{q}_T\bar{T}-  \bar{u}_3(\bar{T},t)\bar{T}+ \sigma_3(\bar{V},\bar{B},\bar{T}) ~dW_3(t),~ \bar{T}(0) = \bar{T}_0.\\
\end{aligned}
\end{equation}
where $dW_i,~ i = 1,2,3$ are one-dimensional Wiener processes, and $\sigma_i,~ i=1,2,3$,  are positive. The expressions for $\sigma_i$ will be obtained using experimental data, as described in Section \ref{sec:results}.

Using a compact notation, we can write \eqref{eq:ItoODE} as 
 \begin{equation}\label{eq:ItoODE_compact}
\begin{aligned}
&\dfrac{d\bX}{dt} =\bF(\bX,\bT, \bU(\bX,t)) + \bS(\bX)~d\bW(t) ,\\
&\bX(0) = \bX_0,
\end{aligned}
\end{equation}
where 
\[
d\bW(t) = 
\begin{pmatrix}
&dW_1(t)\\
&dW_2(t)\\
&dW_3(t)\\
\end{pmatrix}
\]
is a 3-dimensional Wiener process with stochastically independent components and 
\[
\bS(\bx) = 
\begin{pmatrix}
&\sigma_1(\bX) &0 & 0\\
&0&\sigma_2(\bX)&0 \\
&0&0&\sigma_3(\bX) \\
\end{pmatrix}
\]
is the dispersion matrix, with $\sigma_i(\bX) > 0$ for all $\bX$.

We now characterize the state of the stochastic process, describing the evolution of $\bX(t)$ through \eqref{eq:ItoODE_compact}, by its probability density function (PDF). For this purpose, we first assume that the process \eqref{eq:ItoODE_compact} is constrained to stay in a 
bounded convex domain with Lipschitz boundaries, thus $X(t) \in \Omega\subset\R^3_+ = \lbrace x\in\mathbb{R}^3:x_i \geq 0,~ i= 1,2,3 \rbrace$, by virtue of a reflecting barrier on $\partial \Omega$. This is due to the maximum carrying capacity inside a human being. Let $x = (x_1,x_2,x_3)^T$. Now define $f(x,t)$ as the PDF for the stochastic process described by \eqref{eq:ItoODE_compact}, i.e., $f(x,t)$ is the probability of $\bX(t)$ assuming the value $x$ at time $t$. Then, the evolution of the PDF of the process modeled by \eqref{eq:ItoODE_compact} is given through the following Fokker-Planck (FP) equations
\begin{equation}\label{eq:FP}
\begin{aligned}
&\partial_t f(x,t)+\nabla \cdot (\bF(x,\bT,u(x,t))~f(x,t)) = 
\frac{1}{2}\nabla \cdot (\bS^2(\bX) \nabla f(x,t)),\\
&f(x,0) =  f_0(x),
\end{aligned}
\end{equation}
where $f_0(x)$ represents the 
initial PDF distribution that satisfies the following
\begin{equation}\label{eq:initial_cond}
 f_0 \geq 0,\qquad\int_\Omega f_0(x)dx = 1,
 \end{equation}
 The function $f_0(x)$ represents the distribution of the initial state $X_0$ of the process and the domain of definition of the FP problem is $Q=\Omega\times(0,T_f)$, where $T_f$ is the final time of observation. The reflecting barrier conditions assumed on the process correspond to  flux zero boundary conditions for the FP equation \eqref{eq:FP}. For this purpose, we write \eqref{eq:FP} in flux form as 
\begin{equation}\label{eq:FPflux}
\partial_t f(x,t)=\nabla\cdot \mathcal{F}, \qquad f(x,0) =  f_0(x) ,
\end{equation}
where the flux $F$ is given component-wise by 
\begin{equation}\label{eq:flux_def}
\mathcal{F}_j(x,t;f) = \frac{\sigma_j^2(x)}{2}\partial_{x_j} f-\bF_j(x,\bT)f, ~ j = 1,2,3.
\end{equation}
Then, the flux zero boundary conditions can be formulated as follows
\begin{equation}\label{eq:nfbc}
\mathcal{F}\cdot \hat n = 0 \qquad \mbox{ on } \partial\Omega\times(0,T),
\end{equation}
where $\hat n$ is the unit outward normal on $\partial\Omega$. 

We consider $\bT\in U_{ad} = \lbrace y\in \mathbb{R}^6: 0 \leq y_i \leq M_i,~ i=1,\cdots, 6,~ M_i>0\rbrace$, and $u(x,t) = (\bar{u}_1(x_1,t),\bar{u}_2(x_2,t),\bar{u}_3(x_3,t))$ in the admissible set
\[
\bar{u}_i(x_i,t) \in V^i_{ad} = \lbrace \bar{u}_i \in L^2([0,T;H^1(\Omega_i)):\forall (x,t) \in Q,~ 0 \leq u_i(x_i,t) \leq D_i,~ D_i>0\rbrace,~ i = 1,2,3,
\]
where $M_i$ are chosen so as to represent the feasible biological ranges of the parameters, $D_i$ is the maximum tolerable dose for the drug represented by $\bar{u}_i$, and $\Omega_i$ is the one dimensional subdomain in the $i^{th}$ direction.

\begin{remark}

We note that $u$, as a function of $x$, is in $H^1(\Omega)$, represents the fact that the drug dosages are smooth functions of $V,B,T$. This implies that during dose administration, a feasible strategy is to 
increase or decrease the dose feedback in a smooth way depending on the values of $V,B,T$, without severe discontinuities. However, as a function of time, the dosage profiles can have a non-smooth structure, which is why $u$ is $L^2$ with respect to $t$. This is practically motivated due to the fact that the dosage profile is more sensitive and adaptive to $V,B,T$ rather than with respect to $t$ (see e.g., \cite{Gatenby2009}).
\end{remark}

\subsection{FP algorithm for optimal combination therapies}
We now describe the algorithm for obtaining the optimal combination therapies. We follow a three-step process as described below:

\begin{enumerate} 
\item We first estimate the patient specific unknown parameter vector $\bT$, given the values of $f(x,t)$ at specific time instants $t_1,\cdots, t_N$ as $f^*_i(x),~i=1,\cdots,N$. For this purpose, we solve the following optimization problem
\begin{equation}\label{eq:min_problem}
\begin{aligned}
\boldsymbol{\theta}^* = \argmin_{\bT \in U_{ad}} J_1(f,\boldsymbol{\theta}) := &\dfrac{\alpha}{2} \int_Q (f(x,t) - f^*(x,t))^2~dx + \dfrac{\beta}{2}\|\boldsymbol{\theta}\|_{l^2}^2,
\end{aligned}
\end{equation}
subject to the FP system \eqref{eq:FP},\eqref{eq:initial_cond},\eqref{eq:nfbc} with $\bar{u}_i = 0,~ i = 1,2,3$, where $f^*(x,t)$ is the data function formed by interpolating the patient data $f^*_i(x)$.

\item In the next step, we determine the subset of the optimal parameter set $\bT^*$ that is sensitive with respect to the tumor volume $V$. This will be achieved through a global uncertainty and sensitivity analysis using the Latin hypercube sampling-partial rank correlation coefficient method (LHS-PRCC), as described in Section \ref{sec:uncertainty}.

\item Using the information of the sensitive parameters from the previous step, we now decide on the type of drugs to be chosen, and the number of different drugs to be used, represented by the number of $\bar{\alpha}_i \neq 0$. We then formulate a second FP optimization problem as follows:
\begin{equation}\label{eq:second_min_problem}
\begin{aligned}
\min_{\bar{u}_i \in V^i_{ad},\bar{\alpha}_i\neq 0} J_2(\bar{u}_i,f):=\dfrac{\nu}{2} \int_\Omega (f(x,T_f) - f_d(x))^2~dx + \sum_{i=1,
\bar{\alpha}_i \neq 0}^3\dfrac{\beta_i}{2}\int_0^{T_f} \|\bar{u}_i\|_{H^1(\Omega)}^2~dt,
\end{aligned}
\end{equation}
subject to the FP system \eqref{eq:FP},\eqref{eq:initial_cond},\eqref{eq:nfbc},
\end{enumerate}
where $f_d$ is a target PDF at the final time $T_f$.
At the end of Step 3, we not only obtain the types of drugs that can be used for treatment but also the optimal drug concentration and the dosage profile over time. 

\section{Theory of the Fokker-Planck parameter estimation problem}\label{sec:theory}
In this section, we discuss some theoretical results related to FP system \eqref{eq:FP} and  the existence of solutions of the optimization problem \eqref{eq:min_problem}. For this purpose, we denote the FP system  \eqref{eq:FP},\eqref{eq:initial_cond},\eqref{eq:nfbc} as $\cE(f_0,\bT,u)=0$. We also define $V_{ad} = V_{ad}^1 \times V_{ad}^2 \times V_{ad}^3.$ We first discuss the existence of weak solutions of \eqref{eq:FP}.  We have the following proposition.

\begin{proposition}\label{th:proposition1}
Let $f_0 \in H^1(\Omega)$, $f_0\ge 0$, $\bT\in U_{ad}$, and $u \in V_{ad},~ i=1,2,3$. Then, there exists an unique non-negative solution of $\cE(f_0,\bT,u)=0$ given by 
$f \in L^2(0,T;H^1(\Omega)) \cap C([0,T];L^2(\Omega))$.
\end{proposition}

We remark that using classical techniques \cite{Tao},  one can get the $H^2(\Omega)$ regularity in space. Next, because of \eqref{eq:FPflux} and \eqref{eq:nfbc}, we can prove the following proposition that states conservation of the total probability. 
\begin{proposition}\label{th:FPcons}
The FP system given in \eqref{eq:FP},\eqref{eq:initial_cond},\eqref{eq:nfbc} is conservative.
\end{proposition} 
\begin{proof}
Multiplying \eqref{eq:FP} by $\psi \in H^1(\Omega)$, integrating by parts, and using the flux zero boundary conditions \eqref{eq:nfbc}, we obtain the following
\begin{equation}\label{eq:L2_inner_prod}
\begin{aligned}
\int_\Omega \frac{\partial f}{\partial t} \psi dx 
&= - \frac{1}{2} \int_\Omega \bS^2\nabla f \cdot \nabla \psi dx + \int_\Omega (\bF f) \cdot \nabla \psi \, dx,\\
&= - \frac{1}{2} \int_\Omega \bS^2\nabla f \cdot \bS\nabla \psi dx + \int_\Omega (\bF f) \cdot \nabla \psi \, dx. 
\end{aligned}
\end{equation}
Choosing $\psi=1$, we obtain $\int_\Omega f(x,t)dx= \int_\Omega f_0(x)dx = 1$ for all 
$t \in (0,T]$ and this proves the result.
\end{proof}
The following proposition gives a stability property of our FP system.
\begin{proposition}\label{th:stability}
The solution $f$ of the FP system \eqref{eq:FP},\eqref{eq:initial_cond},\eqref{eq:nfbc} satisfies the following stability estimate
\begin{equation}\label{eq:FPstability}
\| f(t) \|_{L^2(\Omega)} 
\le  \|  f_0 \|_{L^2(\Omega)}  \exp \left(  \|\bS^{-1}\|^2_2 N^2 t \right),
\end{equation}
where $N = \sup_{\Omega\times U_{ad} \times V_{ad}} |\bF(x,\bT,u)|$.
\end{proposition}
\begin{proof}
Choosing $\psi=f(\cdot,t)$ in \eqref{eq:L2_inner_prod}, we have 
\begin{equation}\label{eq:FPstability0}
\frac{\partial }{\partial t} \| f(t) \|^2_{L^2(\Omega)} 
= -  \| \bS \nabla f(t) \|^2_{L^2(\Omega)}  + 2 \int_\Omega (\bF f(t)) \cdot \bS^{-1}\bS \nabla f(t) \, dx.
\end{equation}

To estimate the last term in \eqref{eq:FPstability0}, we use the Young's inequality, 
$2bd \le kb^2 + d^2/k$ with $k = \|\bS^{-1}\|_2$, the $L^2$ matrix norm of $\bS^{-1}$, and obtain the following
$$
\frac{\partial }{\partial t} \| f(t) \|^2_{L^2(\Omega)} 
\le   \|\bS^{-1}\|^2_2 N^2 \|  f(t) \|^2_{L^2(\Omega)}  . 
$$
Using Gronwall's inequality, we arrive at the desired result.
\end{proof}

Next, we state and prove some further properties of the solution to \eqref{eq:FP} that is needed for proving the existence of optimal $\bT^*$ and $u^*$.
For this purpose, we have the following proposition.

\begin{proposition} \label{eq:propXX}
Let $f_0 \in H^1(\Omega)$, $f_0\ge 0$, and $\bT\in U_{ad}$. Then, if 
$f$ is a solution to $\cE(f_0,\bT,u)=0$, the following inequalities hold
\begin{equation}
\|f\|_{L^\infty(0,T;L^2(\Omega))} \le c_1 \|f_0\|_{L^2(\Omega)}, 
\label{eq:ineq2}
\end{equation}
\begin{equation}
\|\partial_t f\|_{L^2(0,T;H^{-1}(\Omega))} \le (c_2 + c_3N) \|f_0\|_{L^2(\Omega)} ,
\label{eq:ineq3}
\end{equation}
where $c_1$, $c_2$, $c_3$ are positive constants and $N$ is defined in Proposition \ref{th:stability}. Further, if $\|\bS^{-1}\|^2_2 > \dfrac{1}{N} $, then the following inequality holds
\begin{equation}
\|f\|_{L^2(0,T;H^1(\Omega))} \le  c_4 \|f_0\|_{L^2(\Omega)} , 
\label{eq:ineq1}
\end{equation}
where $c_4$ is a positive constant depending on $N$.
\end{proposition} 
\begin{proof}
The inequality (\ref{eq:ineq2}) follows from (\ref{eq:FPstability}), with 
$$c_1 = \exp \left(  \|\bS^{-1}\|^2_2 N^2 t \right).
$$
To prove inequality (\ref{eq:ineq3}), we define the dual of the $H^1(\Omega)$ norm, given by $H^{-1}(\Omega)$, as follows
$$
\| \partial_t f\|_{H^{-1}(\Omega)} = \sup_{\substack{\psi \in H_0^1(\Omega)\\ \psi\neq {0}}} \dfrac{\langle \partial_t f, \psi \rangle_{L^2(\Omega)}}{\|\psi\|_{H_0^1(\Omega)}}.
$$
From (\ref{eq:L2_inner_prod}), using (\ref{eq:FPstability}) we get 
$$
\langle \partial_t f, \psi \rangle_{L^2(\Omega)} \leq (c_2+c_3N)\|f_0\|_{L^2(\Omega)} \|\psi\|_{H_0^1(\Omega)},
$$
where 
$$
c_2 = \dfrac{\|\bS\|^2_2}{2},~ c_3 = c_1^2.
$$
To prove (\ref{eq:ineq1}), we first integrate (\ref{eq:FPstability0}) in $(0,T)$ to obtain
\begin{equation*}
\| f(T) \|^2_{L^2(\Omega)} -\| f_0 \|^2_{L^2(\Omega)}
= -  \int_0^T\| \bS \nabla f(t) \|^2_{L^2(\Omega)}~dt  + 2 \int_0^T\int_\Omega (\bF f(t)) \cdot \nabla f(t) \, dxdt. 
\end{equation*}
Using the Young's inequality, we have
\begin{equation*}
 \int_0^T\|  \bS\nabla f(t) \|^2_{L^2(\Omega)}~dt \leq \| f_0 \|^2_{L^2(\Omega)}
+ \int_0^T\Bigg(  N \| f(t) \|^2_{L^2(\Omega)} + N\|\bS^{-1}\|^2_2\| \bS \nabla f(t) \|^2_{L^2(\Omega)}\Bigg)\, dt. 
\end{equation*}
This implies
\begin{equation}\label{eq1}
(N\|\bS^{-1}\|^2_2 - 1) \int_0^T\| \nabla f(t) \|^2_{L^2(\Omega)}~dt \leq \| f_0 \|^2_{L^2(\Omega)}
+ N\int_0^T\| f(t) \|^2_{L^2(\Omega)} \, dt. 
\end{equation}
Adding $(N\|\bS^{-1}\|^2_2 - 1)\int_0^T\| f(t) \|^2_{L^2(\Omega)}~dt$ to (\ref{eq1}) we have the following
\begin{equation}\label{eq2}
(N\|\bS^{-1}\|^2_2 - 1)\int_0^T\Bigg(\| f(t) \|^2_{L^2(\Omega)}+\| \nabla f(t) \|^2_{L^2(\Omega)}\Bigg)~dt \leq \| f_0 \|^2_{L^2(\Omega)}
+N\int_0^T\| f(t) \|^2_{L^2(\Omega)} \, dt. 
\end{equation}
Using (\ref{eq:FPstability}), we have
\begin{equation}\label{eq3}
\int_0^T\| f(t) \|^2_{L^2(\Omega)} ~dt
\le \|  f_0 \|^2_{L^2(\Omega)}  \int_0^T \exp \left(  \|\bS^{-1}\|^2_2 N^2 t \right)~dt= 
\dfrac{1}{\|\bS^{-1}\|^2_2 N^2}\Bigg[\exp \left(\|\bS^{-1}\|^2_2 N^2 T \right)-1\Bigg]\|  f_0 \|^2_{L^2(\Omega)}.
\end{equation}
Therefore, we obtain
\begin{equation}
\begin{aligned}
&(N\|\bS^{-1}\|^2_2 - 1)\int_0^T\Bigg(\| f(t) \|^2_{L^2(\Omega)}+\| \nabla f(t) \|^2_{L^2(\Omega)}\Bigg) ~dt\\
&\leq \dfrac{1}{\|\bS^{-1}\|^2_2 N^2}\Bigg[\exp \left(\|\bS^{-1}\|^2_2 N^2 T \right)-1 + \|\bS^{-1}\|^2_2 N^2\Bigg]\|  f_0 \|^2_{L^2(\Omega)}.
\end{aligned}
\end{equation}
This proves (\ref{eq:ineq1}) with $c_4 = \sqrt{\dfrac{1}{(N\|\bS^{-1}\|^2_2 - 1)\|\bS^{-1}\|^2_2 N^2}\Bigg[\exp \left(\|\bS^{-1}\|^2_2 N^2 T \right)-1 + \|\bS^{-1}\|^2_2 N^2\Bigg]}.$
\end{proof}

From the results above, we obtain that the mapping $\Lambda : U_{ad} \times V_{ad} \to C([0,T];H^1(\Omega))$, $(\bT,u) \to f=\Lambda(\bT,u)$ is continuous. Further, using arguments given in \cite{MA}, we can prove that this mapping is also Fr\'echet differentiable. In the next proposition, we discuss some properties of the cost functionals $J_1,J_2$ given in \eqref{eq:min_problem} and \eqref{eq:second_min_problem}, which can be proved using the fact that the PDF $f$ is non-negative.
 
\begin{proposition}\label{proposition3} 
The objective functionals $J_1,J_2$, given in \eqref{eq:min_problem} and \eqref{eq:second_min_problem}, are sequentially weakly lower 
semi-continuous (w.l.s.c.), bounded from below, coercive on $U_{ad},V_{ad}$. respectively, and are Fr\'echet differentiable.  
\end{proposition}

We now state and prove the existence of the optimal parameter set $\bT^*$ and the optimal drug dosage concentration vector $u^*$ in the following theorem.
\begin{theorem}\label{th:existence}
Let $f_0 \in H^1(\Omega)$ satisfy \eqref{eq:initial_cond} and let $J_1,J_2$ be given as in \eqref{eq:min_problem} and \eqref{eq:second_min_problem}. Then, there exists pairs 
$(f_1^*,\bT^*) \in C([0,T];H^1(\Omega)) \times U_{ad}$ and $(f_2^*,u^*) \in C([0,T];H^1(\Omega)) \times V_{ad}$ such that $f_1^*,f_2^*$ are solutions of $\cE(f_0,\bT^*,0)=0,\cE(f_0,\bT^*,u^*)=0$, respectively, and $\bT^*,u^*$ minimize $J_1,J_2$ in $U_{ad},V_{ad}$, respectively. 
\end{theorem}
\begin{proof} 
We first prove the existence of minimizer of $J_1$ in \eqref{eq:min_problem}. Since $J_1$ is bounded below, there exists a minimizing sequence $(\bT^m) \in U_{ad}$. Since $U_{ad}\subset \mathbb{R}^6$, and $J_1$ is sequentially w.l.s.c. as well as coercive in $U_{ad}$, 
this sequence is bounded. Therefore, it contains a convergent 
subsequence $(\bT^{m_l})$ in $U_{ad}$ such that $u^{m_l} \rightarrow \bT^*$. 
Correspondingly, the sequence $(f^{m_l})$, where $f^{m_l}=\Lambda(\bT^{m_l},0)$, 
is bounded in $L^2(0,T; H^1(\Omega) )$ by \eqref{eq:ineq1}, while the 
sequence of the time derivatives, $(\partial_t f^{m_l})$, 
is bounded in $L^2(0,T; H^{-1}(\Omega))$ by \eqref{eq:ineq3}. Therefore, both the sequences 
converge weakly to $f_1^*$ and $\partial_t f_1^* $, respectively. From the above discussion, we obtain weak convergence of the sequence $(\bF(\bT^{m_k},0)f^{m_l})$ in $L^2(0,T,L^2(\Omega))$. It now follows that $f_1^*=\Lambda(\bT^*,0)$, 
and the pair $(f_1^*,\bT^*)$ minimizes $J_1$.

For proving existence of a minimizer of $J_2$, given in \eqref{eq:second_min_problem}, we can follow the same arguments as above noting the fact that $V_{ad}$ being a closed subspace of a Hilbert space and $J_2$ being coercive in $V_{ad}$ yields a convergent subsequence $(u_{m_l})$ of a minimizing sequence $(u_m)$ for $J_2$, and the compactness result of Aubin-Lions \cite{Lions1969} yields strong 
convergence of a subsequence $(f^{m_k})$ of a sequence $(f^{m_l} = \Lambda(\bT^*,u_{m_l}))$ in $L^2(0,T,L^2(\Omega))$. 
\end{proof}
We now introduce the following reduced functionals
\begin{equation}\label{reduced_functional}
\hat{J}_1(\bT)=J_1(\Lambda(\bT,0),\bT),~\hat{J}_2(u)=J_2(\Lambda(\bT^*,u),u) .
\end{equation} 
The following proposition shows the differentiability of the reduced functionals $\hat{J}_1,\hat{J}_2$ that can be proved using similar arguments as in \cite{trobuch}.
 
\begin{proposition}\label{proposition6}
The reduced functionals $\hat{J}_1(\bT),\hat{J}_2(u)$ is differentiable, and their derivatives are given by
$$
\begin{aligned}
&d\hat J_1(\bT) \cdot \bps_1= \Big\langle{ \beta \bT - \int_Q\nabla_{\bT} \bF\cdot \nabla p_1~dxdt , \bps_1 }\Big\rangle_{L^2}, \qquad \forall \bps_1 \in U_{ad},\\
&(d\hat J_2(u) \cdot \bps_2)_i= \Big\langle{ \mu(\Omega_j)\mu(\Omega_k)(\beta_i \bar{u}_i - \beta_i \Delta_i \bar{u}_i) -\int_{\Omega_j}\int_{\Omega_k}\alpha_ix_i f\cdot \nabla_{x_i} p_2~dx_jdx_k , (\bps_2)_i }\Big\rangle_{L^2}, \\
 &\hspace{100mm}i,j,k = 1,2,3,~ j,k\neq i,\qquad \forall \bps_2 \in V_{ad},
\end{aligned}
$$
where $p_1$ is the solution to the adjoint equation 
\begin{equation*}
\begin{aligned}
-\partial_t p_1(x,t)-f(x,t)(\bF(x,\bT,0)\cdot \nabla p_1(x,t)) - &
\frac{1}{2}\nabla \cdot (\bS^2(x) \nabla p_1(x,t)) = -\alpha (f(x,t)-f^*(x,t)), ~ \mbox{in } \Omega \times(0,T_f)\\
\frac{\partial{p_1}}{\partial{n}}  = 0, \qquad & \mbox{ on } \partial\Omega\times(0,T_f),
\end{aligned}
\end{equation*}
with $p_1(x,T)=0$ and $f$ satisfying $\cE(f_0,\bT,0)  = 0$, and 
$p_2$ is the solution to the adjoint equation 
\begin{equation*}
\begin{aligned}
-\partial_t p_2(x,t)-f(x,t)(\bF(x,\bT^*,u)\cdot \nabla p_2(x,t)) - &
\frac{1}{2}\nabla \cdot (\bS^2(x) \nabla p_2(x,t)) = 0,~ \mbox{in } \Omega\times(0,T_f)\\
\frac{\partial{p}_2}{\partial{n}}  = -\nu (f(x,T_f)-f_d(x)), \qquad & \mbox{ on } \partial\Omega\times(0,T_f),
\end{aligned}
\end{equation*}
with $p_2(x,T)=0$ and $f$ satisfying $\cE(f_0,\bT^*,u)  = 0$.
\end{proposition}

The optimality conditions corresponding to the minimization problem \eqref{eq:min_problem} can now be written as

\begin{equation}\label{eq:opt_FP}
\begin{aligned}
&\partial_t f(x,t)+\nabla \cdot (\bF(x,\bT,0)~f(x,t)) = 
\frac{1}{2}\nabla \cdot (\bS^2(x) \nabla f(x,t)), \qquad \mbox{ in } \Omega\times(0,T_f),\\
&f(x,0) =  f_0(x), \qquad \mbox{ in } \Omega,\\
&\mathcal{F}\cdot \hat n = 0, \qquad \mbox{ on } \partial\Omega\times(0,T_f).
\end{aligned}\tag{FOR1}
\end{equation}

\begin{equation}\label{eq:opt_adj}
\begin{aligned}
-\partial_t p_1(x,t)-f(x,t)(&\bF(x,\bT,0)\cdot \nabla p_1(x,t)) -
\frac{1}{2}\nabla \cdot (\bS^2(x) \nabla p_1(x,t)) = -\alpha (f(x,t) - f^{*}(x,t)), ~\mbox{ in } \Omega\times(0,T_f),\\
& p_1(x,T)=0, \qquad \mbox{ in } \Omega,\\
&\frac{\partial{p}_1}{\partial{n}}  = 0, \qquad  \mbox{ on } \partial\Omega\times(0,T_f).\\
\end{aligned}\tag{ADJ1}
\end{equation}

\begin{equation}\label{eq:opt_cond}
\Big\langle{ \beta \bT - \int_Q\nabla_{\bT} \bF\cdot \nabla p_1~dxdt , \bps_1 }\Big\rangle_{L^2} \geq 0,\qquad \forall \bps \in U_{ad}. \tag{OPT1}
\end{equation}

The optimality conditions corresponding to the minimization problem \eqref{eq:second_min_problem} can be written as

\begin{equation}\label{eq:opt_FP2}
\begin{aligned}
&\partial_t f(x,t)+\nabla \cdot (\bF(x,\bT^*,u)~f(x,t)) = 
\frac{1}{2}\nabla \cdot (\bS^2(x) \nabla f(x,t)), \qquad \mbox{ in } \Omega\times(0,T_f),\\
&f(x,0) =  f_0(x), \qquad \mbox{ in } \Omega,\\
&\mathcal{F}\cdot \hat n = 0, \qquad \mbox{ on } \partial\Omega\times(0,T_f).
\end{aligned}\tag{FOR2}
\end{equation}

\begin{equation}\label{eq:opt_adj2}
\begin{aligned}
-\partial_t p_2(x,t)-f(x,t)(&\bF(x,\bT^*,u)\cdot \nabla p_2(x,t)) -
\frac{1}{2}\nabla \cdot (\bS^2(x) \nabla p_2(x,t)) = 0, ~\mbox{ in } \Omega\times(0,T_f),\\
& p_2(x,T)= -\nu (f(x,T_f)-f_d(x)), \qquad \mbox{ in } \Omega,\\
&\frac{\partial{p}_2}{\partial{n}}  = 0, \qquad  \mbox{ on } \partial\Omega\times(0,T_f).\\
\end{aligned}\tag{ADJ2}
\end{equation}

\begin{equation}\label{eq:opt_cond2}
\begin{aligned}
&\Big\langle{ \mu(\Omega_j)\mu(\Omega_k)(\beta_i \bar{u}_i - \beta_i \Delta_i \bar{u}_i) -\int_{\Omega_j}\int_{\Omega_k}\alpha_ix_i f\cdot \nabla_{x_i} p_2~dx_jdx_k , (\bps_2)_i }\Big\rangle_{L^2} \geq 0, \\
 &\hspace{60mm}i,j,k = 1,2,3,~ j,k\neq i,\qquad \forall \bps_2 \in V_{ad}. 
 \end{aligned}
 \tag{OPT2}
\end{equation}

\section{Numerical discretization schemes for solving the FP optimality system}\label{sec:numerical_FP}

\subsection{Discretization of the forward and adjoint FP equations}

In this section, we describe the numerical discretization schemes for solving the forward and adjoint FP equations given in \eqref{eq:opt_FP}- \eqref{eq:opt_adj} and \eqref{eq:opt_FP2}- \eqref{eq:opt_adj2} . For this purpose, we consider a sequence of uniform grids $\lbrace\Omega_h\rbrace_{h>0}$ given by  
$$
\Omega_h = \lbrace(x_1,x_2,x_3)\in\mathbb{R}^3:(x_{1i},x_{2j},x_{3k}) = (x_{10}+ih,x_{20}+jh,x_{30}+kh)\rbrace,
$$
where $(i,j,k)\in
\lbrace{0,\hdots,N_{x_1}}\rbrace\times \lbrace{0,\hdots,N_{x_2}}\rbrace\times \lbrace{0,\hdots,N_{x_3}}\rbrace\cap\Omega$ and $N_{x_i}$ represents the number of grid points along the $i^{th}$ coordinate direction. We also define $\delta{t}=T/N_t$ to be the time step, where $N_t$ denotes the maximum number of time steps. With this setting, we now consider the discretized domain for $\Omega$ as follows
$$
Q_{h,\delta{t}} = \lbrace{(x_{1i},x_{2j},x_{3k},t_m):\, (x_{1i},x_{2j},x_{3k})\in\Omega_h,~t_m=m\delta{t},~0\leq m\leq N_t}\rbrace.
$$
We denote the value of $f(x,t)$ on the discrete domain $Q_{h,\delta{t}}$ as $f_{i,j}^m$. 

For the spatial discretization, we will use the Chang-Cooper (CC) scheme \cite{CC}, which is represented by the following discretization of the flux term in \eqref{eq:opt_FP} at time $t_m$ 
$$
\nabla\cdot{\mathcal{F}}=\frac{1}{h}\left[(\mathcal{F}_{i+\frac{1}{2},j,k}^m-\mathcal{F}_{i-\frac{1}{2},j,k}^m)+(\mathcal{F}_{i,j+\frac{1}{2},k}^m-\mathcal{F}_{i,j-\frac{1}{2},k}^m) + (\mathcal{F}_{i,j,k+\frac{1}{2}}^m-\mathcal{F}_{i,j,k-\frac{1}{2}}^m) \right],
$$
where $\mathcal{F}_{i+\frac{1}{2},j,k}^m,~\mathcal{F}_{i,j+\frac{1}{2},k}^m,~\mathcal{F}_{i,j,k+\frac{1}{2}}^m$ represent the numerical flux in the $i,j,k$ directions, respectively, at the point $(x_{1i},x_{2j},x_{3k})$. The numerical flux $\mathcal{F}_{i+\frac{1}{2},j,k}^m$ in the $i^{th}$ direction is given as follows
\begin{equation}\label{eq:numflux}
\mathcal{F}_{i+\frac{1}{2},j,k}^m = \left[{(1-\delta_i)B_{i+\frac{1}{2},j,k,m} + \frac{\sigma_i^2}{2h}}\right]f_{i+1,j}^{m} - 
\left[{\frac{\sigma_i^2}{2h}-\delta_iB_{i+\frac{1}{2},j,k,m} }\right]f_{i,j}^{m},
\end{equation}
where
\begin{equation}\label{eq:coefB}
B_{i+\frac{1}{2},j,m	} = -\bF_1(x_{1i+\frac{1}{2}},x_{2j},x_{3k},\bT,u),
\end{equation}
and 
\begin{equation}\label{eq:coefd}
\begin{aligned}
&\delta_i = \frac{1}{w_{i+\frac{1}{2},j}^m}-\frac{1}{\exp(w_{i+\frac{1}{2},j,k}^m)-1}, 
\qquad w_{i+\frac{1}{2},j,k}^m = 2hB_{i+\frac{1}{2},j,k}/\sigma_i^2.\\
\end{aligned}
\end{equation}
A similar formulae also holds true for the fluxes in the other directions.

For discretizing the time derivative, we will use the Douglas-Gunn (D-G) scheme. The D-G scheme is a three-step method that gives a consistent discretization of the FP equation at each step. At every step, the scheme is implicit in one direction only that results in a simpler system to solve. The D-G scheme is coupled with the CC scheme that results in a fully discretized scheme for solving the FP equation \eqref{eq:opt_FP}. We call this scheme as the DG3-CC scheme. Below, we describe the formulation of the fully discrete DG3-CC scheme. We introduce an auxiliary time steps $t_{m^*},t_{m^{**}}$. For notational convenience, we only use indices in the flux $\mathcal{F}$ that represent the flux in the corresponding direction and drop the other indices. For e.g., $\mathcal{F}_{i+\frac{1}{2}}$ represents $\mathcal{F}_{i+\frac{1}{2},j,k}$, the flux in the $i^{th}$ direction.
\begin{equation}\label{eq:DG4CC}
\begin{aligned}
\frac{f_{i,j,k} ^{m^*}-f_{i,j,k} ^{m}}{\delta{t}} &= \frac{1}{2h}(\mathcal{F}_{i+\frac{1}{2}}^{m^*}-\mathcal{F}_{i-\frac{1}{2}}^{m^*})\\
&+\frac{1}{2h}(\mathcal{F}_{i+\frac{1}{2}}^{m}-\mathcal{F}_{i-\frac{1}{2}}^{m})
+\frac{1}{h}(\mathcal{F}_{j+\frac{1}{2}}^{m}-\mathcal{F}_{j-\frac{1}{2}}^{m}) +
\frac{1}{h}(\mathcal{F}_{k+\frac{1}{2}}^{m}-\mathcal{F}_{k-\frac{1}{2}}^{m}),\\
\frac{f_{i,j,k} ^{m^{**}}-f_{i,j,k} ^{m}}{\delta{t}} &=\frac{1}{2h}(\mathcal{F}_{i+\frac{1}{2}}^{m^*}-\mathcal{F}_{i-\frac{1}{2}}^{m^*})+\frac{1}{2h}(\mathcal{F}_{j+\frac{1}{2}}^{m^{**}}-\mathcal{F}_{j-\frac{1}{2}}^{m^{**}})\\
&+\frac{1}{2h}(\mathcal{F}_{i+\frac{1}{2}}^{m}-\mathcal{F}_{i-\frac{1}{2}}^{m})
+\frac{1}{2h}(\mathcal{F}_{j+\frac{1}{2}}^{m}-\mathcal{F}_{j-\frac{1}{2}}^{m}) +
\frac{1}{h}(\mathcal{F}_{k+\frac{1}{2}}^{m}-\mathcal{F}_{k-\frac{1}{2}}^{m}), \\
\frac{f_{i,j,k} ^{m+1}-f_{i,j,k} ^{m}}{\delta{t}} &=\frac{1}{2h}(\mathcal{F}_{i+\frac{1}{2}}^{m^*}-\mathcal{F}_{i-\frac{1}{2}}^{m^*})+\frac{1}{2h}(\mathcal{F}_{j+\frac{1}{2}}^{m^{**}}-\mathcal{F}_{j-\frac{1}{2}}^{m^{**}})+
\frac{1}{2h}(\mathcal{F}_{k+\frac{1}{2}}^{m+1}-\mathcal{F}_{k-\frac{1}{2}}^{m+1})\\
&+\frac{1}{2h}(\mathcal{F}_{i+\frac{1}{2}}^{m}-\mathcal{F}_{i-\frac{1}{2}}^{m})
+\frac{1}{2h}(\mathcal{F}_{j+\frac{1}{2}}^{m}-\mathcal{F}_{j-\frac{1}{2}}^{m}) +
\frac{1}{2h}(\mathcal{F}_{k+\frac{1}{2}}^{m}-\mathcal{F}_{k-\frac{1}{2}}^{m}), \\
\end{aligned}
\end{equation}
with the initial condition $f^0_{i,j,k} = f_0(x_{1i},x_{2j},x_{3k})$,
for all $(i,j,k)\in \lbrace 1,\hdots,N_x-1\rbrace$. The flux zero boundary conditions in the $i^{th}$ direction is given as follows 
\begin{equation}\label{nfdis}
\begin{aligned}
&\mathcal{F}(N_x-1/2,j,k)=0,~ \mathcal{F}(1/2,j,k) = 0, \qquad \forall j,k=0,\hdots , N_x.\\
\end{aligned}
\end{equation}
A similar condition holds for flux zero boundary condition in the other directions.
We now analyze some properties of the DG3-CC scheme \eqref{eq:DG4CC}-\eqref{nfdis}. The following lemma shows that the DG3-CC scheme is conservative.
\begin{lemma}\label{lemma:conservative}
The DG3-CC scheme (\ref{eq:DG4CC})-(\ref{nfdis}) is conservative in the discrete sense.
\end{lemma}
\begin{proof} 
Summing over all $i,j$ in the last equation of \eqref{eq:DG4CC}, we obtain
\begin{equation}\label{sum}
\begin{aligned}
\sum_{i,j,k}\frac{f_{i,j,k} ^{m+1}-f_{i,j,k} ^{m}}{\delta{t}} =\sum_{i,j,k}& \Bigg[\frac{1}{2h}(\mathcal{F}_{i+\frac{1}{2}}^{m^*}-\mathcal{F}_{i-\frac{1}{2}}^{m^*})+\frac{1}{2h}(\mathcal{F}_{j+\frac{1}{2}}^{m^{**}}-\mathcal{F}_{j-\frac{1}{2}}^{m^{**}})+
\frac{1}{2h}(\mathcal{F}_{k+\frac{1}{2}}^{m+1}-\mathcal{F}_{k-\frac{1}{2}}^{m+1})\\
+&
+\frac{1}{2h}(\mathcal{F}_{i+\frac{1}{2}}^{m}-\mathcal{F}_{i-\frac{1}{2}}^{m})
+\frac{1}{2h}(\mathcal{F}_{j+\frac{1}{2}}^{m}-\mathcal{F}_{j-\frac{1}{2}}^{m})
+
\frac{1}{2h}(\mathcal{F}_{k+\frac{1}{2}}^{m}-\mathcal{F}_{k-\frac{1}{2}}^{m}) \Bigg].
\end{aligned}
\end{equation}
The right hand side of (\ref{sum}) is a telescoping series and, thus, we have
\begin{equation}
\begin{aligned}
\sum_{i,j}\frac{f_{i,j,k} ^{m+1}-f_{i,j,k} ^{m}}{\delta{t}} =0.
\end{aligned}
\end{equation}
This gives us
\begin{equation}\label{eq:conserv}
\sum_{i,j,k}f_{i,j,k} ^{m+1}=\sum_{i,j,k}f_{i,j,k} ^{m},\qquad \forall m = 0,\hdots, N_t-1,
\end{equation}
which proves that the DG3-CC scheme is conservative in the discrete sense.
\end{proof}

Next, we show the positivity of the DG-CC scheme, i.e. $f^0 \geq 0 $ implies $f^m \geq 0$ for all $m >0$. For this purpose, we assume that $\bF$ is Lipschitz continuous with Lipshitz constant $\Gamma$ independent of $t$, i.e., 
\begin{equation}\label{eq:Lipschitz}
\|\bF(x,t) - \bF(y,t)\| \leq \Gamma \|x-y\|,\qquad \forall x,y \in \Omega,~ t\in [0,T].
\end{equation} 
Such a condition also ensures unique solvability of the underlying ODE system \eqref{eq:ODE_compact}.
Then, we can use similar arguments as in \cite[Th. 4.1]{Roy2018a} to obtain the following result
\begin{theorem}\label{th:positivity}
The DG3-CC scheme is positive under the Courant-Friedrichs-Lewy (CFL)-like condition 
\begin{equation}\label{CFL}
\delta t < \min \left \lbrace \dfrac{2}{\Gamma}, \dfrac{2h^2}{V} \right \rbrace,
\end{equation}
where $\Gamma$ is the Lipschitz constant given in \eqref{eq:Lipschitz} and 
\begin{equation}\label{value_V}
\begin{aligned}
V = \dfrac{h \underline B}{e^{2h\underline B/\overline C}-1}+
\dfrac{h \overline B}{1-e^{-2h\overline B/\overline C}}, \qquad
\mbox{where } &\underline B=\min_{x,t} \lbrace \bF(x,t)\rbrace,\; \overline B=\max_{x,t} \lbrace\bF(x,t)\rbrace, ~\overline C = \max_{i}\sigma_i^2
\end{aligned}
\end{equation}
\end{theorem}
Next, we state a discrete stability property of the DG3-CC scheme that can be proved using similar arguments as in \cite[Th. 4.3]{Roy2018a}.

\begin{theorem}\label{th:stability1}
The solution $f_{i,j,k}^m$ obtained using the DG3-CC scheme for the FP equation \eqref{eq:opt_FP}  with a source $g(x,t)$, under the CFL-like condition (\ref{CFL}), satisfies the following $L^1$ stability result
$$
\|f^{m}\|_1 \leq  \|f^0\|_1 + \delta{t} \sum_{n=0}^m \max(\|g^n\|_1, \|g^{n-1/2}\|_1), 
\qquad m = 0,\hdots N_t-1,
$$
where $\| \cdot \|_1$ is the discrete $L^1$ norm.
\end{theorem}
We now analyze the consistency properties of the DG3-CC scheme. For this purpose, we note that the the DG3-CC scheme given in \eqref{eq:DG4CC}-\eqref{nfdis} can be written in one step as
we obtain
\begin{equation}\label{eq:DG_comb}
\begin{aligned}
\frac{f_{i,j,k} ^{m+1}-f_{i,j,k} ^{m}}{\delta{t}}&= \frac{1}{2h}\left[{(\bar{D}^{1}_{x} + \bar{D}^{2}_{x} + \bar{D}^{3}_{x})(f_{i,j,k}^m+f_{i,j,k}^{m+1})}\right]\\
&-\frac{\delta{t}}{4h^2}\left[{(\bar{D}^1_x\bar{D}^2_x + \bar{D}^1_x\bar{D}^3_x  + \bar{D}^2_x\bar{D}^3_x )(f_{i,j,k}^{m+1}-f_{i,j,k}^m)}\right]\\
&+\frac{\delta{t}^2}{8h^3}\left[{(\bar{D}^1_x\bar{D}^2_x\bar{D}^3_x  )(f_{i,j,k}^{m+1}-f_{i,j,k}^m)}\right],
\end{aligned}
\end{equation}
where 
\begin{equation}
\begin{aligned}
\bar{D}^1_x f_{i,j,k}^m = &D_+C_{i-\frac{1}{2},j,k}^{m^*}D_-f_{i,j,k}^{m} + D_+B_{i-\frac{1}{2},j,k}^{m^*}M_\delta f_{i,j,k}^{m},\\
\bar{D}^2_x f_{i,j,k}^m = &D_+C_{i,j-\frac{1}{2},k}^{m^{**}}D_-f_{i,j,k}^{m} + D_+B_{i,j-\frac{1}{2},k}^{m^{**}}M_\delta f_{i,j,k}^{m},\\
\bar{D}^3_x f_{i,j,k}^m = &D_+C_{i,j,k-\frac{1}{2},l}^{m^{***}}D_-f_{i,j,k}^{m} + D_+B_{i,j,k-\frac{1}{2},l}^{m^{***}}M_\delta f_{i,j,k}^{m},\\
\end{aligned}
\end{equation}
and for an index $r\in \lbrace i,j,k\rbrace$
\begin{equation*}
\begin{aligned}
&D_+ f_r = (f_{r+1}-f_{r})/h,\\
&D_- f_r = (f_{r}-f_{r})/h,\\
&M_{\delta}f_r = (1-\delta_{r-1})f_{r} + \delta_{r-1}f_{r-1}.\\
\end{aligned}
\end{equation*}

The first term on the right hand side of \eqref{eq:DG_comb} corresponds to the Crank-Nicholson method with CC discretization for the spatial operator. Using similar arguments as in \cite[Lemma 3.2 and Th. 3.6]{MB}, using a Taylor series expansion, one can show that under the CFL condition (\ref{CFL}), the truncation error corresponding to the first term on the right hand side of \eqref{eq:DG_comb}  is $\mathcal{O}(\delta{t}^2+h^2)$. For the other two terms, using similar arguments as in \cite[Lemma 3.2 and Th. 3.6]{MB} and \cite[Lemma 4.2]{Roy2018a}, one can show that the truncation error is $\mathcal{O}(\delta{t}^2+\delta t^2 h^2)$. Defining the overall truncation error as 
\begin{equation}\label{trunc_error}
\begin{aligned}
\pmb\varphi_{i,j,k}^{m+1} &:=
\frac{f_{i,j,k} ^{m+1}-f_{i,j,k} ^{m}}{\delta{t}}- \frac{1}{2h}\left[{(\bar{D}^{1}_{x} + \bar{D}^{2}_{x} + \bar{D}^{3}_{x})(f_{i,j,k}^m+f_{i,j,k}^{m+1})}\right]\\
&+\frac{\delta{t}}{4h^2}\left[{(\bar{D}^1_x\bar{D}^2_x + \bar{D}^1_x\bar{D}^3_x  + \bar{D}^2_x\bar{D}^3_x )(f_{i,j,k}^{m+1}-f_{i,j,k}^m)}\right]-\frac{\delta{t}^2}{8h^3}\left[{(\bar{D}^1_x\bar{D}^2_x\bar{D}^3_x )(f_{i,j,k}^{m+1}-f_{i,j,k}^m)}\right],
\end{aligned}
\end{equation}
we obtain the following result for the truncation error estimate of the DG3-CC scheme
\begin{lemma}\label{th:truncation}
The truncation error (\ref{trunc_error}) of the DG3-CC scheme (\ref{eq:DG4CC})-(\ref{nfdis}) is of order $\mathcal{O}(\delta{t}^2+h^2)$ under the CFL-like condition (\ref{CFL}).
\end{lemma}
Using Lemma \ref{th:truncation}, Theorem \ref{th:stability1} and arguments as in \cite{Roy2016,Roy2018a}, we obtain the following convergence error estimate of the DG3-CC scheme
\begin{theorem}\label{convergence_DG4CC}
The DG3-CC scheme \eqref{eq:DG4CC}-\eqref{nfdis} is convergent with an error of order $\mathcal{O}(\delta{t}^2+h^2)$ under the CFL condition (\ref{CFL}) in the discrete $L^1$ norm. 
\end{theorem}
For the adjoint equations \eqref{eq:opt_adj} and \eqref{eq:opt_adj2}, we use the D-G scheme for the time discretization in the first term, one sided finite difference discretization for the second term, and central difference for the third term on the left hand side of \eqref{eq:opt_adj} and \eqref{eq:opt_adj2}.

\subsection{A projected NCG optimization scheme\label{sec:optimization}}

For solving the optimization problems \eqref{eq:min_problem} and \eqref{eq:second_min_problem}, we use a projected non-linear conjugate gradient (NCG) scheme (see for e.g., \cite{MA,Pal2020,Pal2021,Roy2016,Roy2018a}). It falls under the class of non-linear optimization schemes, where the objective functional nonlinear yet differentiable with respect to the optimization variables. The NCG scheme has been used to solve several finite and infinite dimensional optimization problems and has been demonstrated to provide fast and accurate solutions of the optimality system, even for finite dimensional optimization problems (for e.g., see the discussion in \cite{MA,Pal2020,Pal2021}. For non-linear optimization schemes involving non-differentiable objective functionals, one can use proximal methods or semi-smooth Newton schemes (see for e.g., \cite{Gupta1,Gupta2,Roy2020}). To describe the NCG scheme for solving the minimization problems \eqref{eq:min_problem} and \eqref{eq:second_min_problem}, we generically denote the reduced functional corresponding to either of the minimization problems as $\hat{J}$, and the associated optimization variable as $\bP$. We start with an initial guess $\bP_0$ for the optimization problem and, correspondingly, calculate 
$$
d_0=g_0:= \nabla_{\bP}\hat{J}(\bP_0),
$$
where $\nabla_{\bP}\hat{J}$ is given by \eqref{eq:opt_cond} or \eqref{eq:opt_cond2}. The search directions are then obtained recursively as follows
\begin{equation}
d_{k+1} = -g_{k+1}+\beta_kd_k,
\end{equation}
where $g_k=\nabla\hat{J}(u_k),~k=0,1,\hdots$. The parameter $\beta_k$ is chosen according to the formula of Hager-Zhang \cite{hag:zha}  given by
\begin{equation}\label{eq:HG}
\beta_k^{HG} = \frac{1}{d_k^Ty_k}\left({y_k-2d_k\frac{\|y_k\|_{l^2}^2}{d_k^Ty_k}}\right)^Tg_{k+1},
\end{equation}
where $y_k = g_{k+1}-g_k$.
Next, a conjugate gradient descent scheme is used as follows to update the optimization variable iterate 
\begin{equation}\label{globGradOpt}
\bP_{k+1} = \bP_k + \alpha_k \,  d_k ,
\end{equation}
where $k$ is an index of the iteration step and $\alpha_k >0$ is a steplength obtained using a backtracking line search algorithm. We use the following Armijo condition of sufficient decrease of $\hat{J}$ for the backtracking line search
\begin{equation}\label{arm}
\hat{J}(\bP_k+\alpha_kd_k)\leq \hat{J}(\bP_k)+\delta\alpha_k\langle\nabla_{\bP}\hat{J}(\bP_k),d_k\rangle_{L^2},
\end{equation}
where $0 < \delta < 1/2$ and the scalar product $\langle u,v\rangle_{L^2} $ is the discrete $l^2$ inner product in $\R^6$ for the minimization problem \eqref{eq:min_problem}, and represents the standard $L^2([0,T;H^1(\Omega))^3$ inner product for the minimization problem \eqref{eq:second_min_problem}. Finally, the gradient update step is combined with a projection step onto the admissible sets in the following way
\begin{equation}\label{eq:localGradOpt}
\bP_{k+1} 
= P_{U}\left[ \bP_k + \alpha_k \, d_k \right] ,
\end{equation}
where 
$$
P_{U}\left[\bP\right] = \left(\max\lbrace 0,\min \lbrace N_i, \bP_i\rbrace\rbrace,~ \forall i = 1, \cdots, s \right),
$$
with $U = U_{ad}$ or $V_{ad}$, $s = 6$ or 3 and $N_i = M_i$ or $D_i$, corresponding to the minimization problems \eqref{eq:min_problem} and \eqref{eq:second_min_problem}, respectively.
The projected NCG scheme can be summarized in the following algorithm:

\begin{algorithm}[Projected NCG Scheme]\label{algo:algo1}\ 
\begin{enumerate}
\item Input: initial approx. $\bP_0$. Evaluate $d_0 = -\nabla_{\bP}\hat{J}(\bP_0)$, index $k=0$, maximum $k=k_{max}$, 
tolerance =$tol$.
\item While $(k<k_{max})$ do
\item Set $\bP_{k+1} = P_{U}\left[ \bP_k + \alpha_k \, d_k \right]$, where $\alpha_k$ is obtained using a line-search algorithm.
\item Compute $g_{k+1} = \nabla_{\bP}\hat{J}(\bT_{k+1})$.
\item Compute $\beta_k^{HG}$ using (\ref{eq:HG}).
\item Set $d_{k+1}=-g_{k+1}+\beta_k^{HG}d_k$.
\item If $\|\bP_{k+1}-\bP_k\|_{l^2} < tol$, terminate.
\item Set $k=k+1$.
\item End while.
\end{enumerate}
\end{algorithm}

\section{Global uncertainty and sensitivity analysis of optimal parameter set} \label{sec:uncertainty}

Once we obtain the optimal parameter sets, we next want to determine the most sensitive parameters with respect to the tumor volume. In this context, it should be noted that the patient-specific parameters can be considered as random variables due to the uncertainties in experimental data \cite{Par49}. It is well-known that any uncertainty in the chosen parameter values may result in inconsistency when it comes to the model's prediction of resulting dynamics. Also, the degree of uncertainty guides the significance of the inconsistency introduced \cite{Helton50}. As such, uncertainty analysis should be used as a tool to quantify the uncertainty in the model output that is a result of the uncertainty in the input parameters. Now, sensitivity analysis, which naturally follows uncertainty analysis, helps in assessing how the overall inconsistency in the model output can be attributed to different input sources. Taken together, in context of accurate assessment of treatments, uncertainty and sensitivity analyses aims to perform the following: (i) identify the key patient-specific parameters, among all input parameters, whose sensitivity significantly contribute to the tumor volume and (ii) rank the identified parameters depending on how much they contribute to this sensitivity.

Although one may carry out a local sensitivity analysis, where the sensitivity of one parameter is studied separately by keeping rest of the parameters fixed at their baseline values, such a method may not be accurate in assessing uncertainties \cite{Hoare51}. Hence, we propose a multi-dimensional parameter space globally that allows all uncertainties to be identified simultaneously. To facilitate this, we employ two efficient statistical tools - Latin hypercube sampling (LHS) method and partial rank correlation coefficient (PRCC) analysis; see \cite{Marino52}. Briefly, in the LHS method, we start six uncertain patient-specific parameter set $\bT$, that are associated with the mathematical model under study. Then, we use Monte Carlo simulation technique to generate $M$ random numbers for each of the six uncertain parameters to produce a $(M\times 6)$ matrix. We call this matrix as the LHS matrix. As a thumb rule, We choose $M$ such that $M > (\frac{4}{3})k$. Note that each row of the LHS matrix can be used as an input vector to generate the uni-dimensional output measure. Hence, we generate $M$ different values for the output measure. The output variable is thus a vector with dimension $(M\times 1)$. Finally, we compute the PRCC between each uncertain parameter and the output variable to identify the parameters that have significant PRCC values. Below, we summarize the steps involved in the LHS-PRCC analysis.

\begin{algorithm}[LHS-PRCC Scheme]\label{algo:algo2}\ 
\begin{enumerate}

\item[Step 1:] For each uncertain parameter in the set $\bT$, specify a PDF. In this way, the variability in the PDF becomes a direct measure of the variability of the uncertain parameter.

\item[Step 2:] To ensure that the sampling distribution of the values for each uncertain parameter adequately reflects the shape of the chosen PDF, divide each PDF into $M$ equi-probable and non-overlapping intervals. 

\item[Step 3:] Randomly draw a number from each interval corresponding to each uncertain parameter exactly once to make sure that the entire range for each parameter is explored. In this context, the drawings are done independently for each parameter. This results in $M$ different values for each of the $k$ uncertain parameters.

\item[Step 4:] Create a LHS matrix with dimension $(M\times 6)$ using the values generated in Step 3. In this matrix, the numbers in each column are not arranged in any particular order. Thus, each row of the LHS matrix represents six random numbers with each random number representing a particular uncertain parameter.

\item[Step 5:] Using each row of the LHS matrix obtained in Step 4, compute the tumor volume $V$, known as the output measure. This results in $M$ different values of the output measure, noting that there are $M$ rows in the LHS matrix. Call this as the output vector having dimension $(M\times 1)$.

\item[Step 6:] Rank transform the LHS matrix, i.e., transform the values in each column of the LHS matrix to ranks. Denote the resulting ranked LHS matrix as $X_R=[X_{1R},X_{2R},\cdots,X_{kR}]$. Note that each $X_{iR},$ $i=1,\cdots,6,$ represents the rank transform of the $i$-th uncertain parameter. Similarly, rank transform the output vector and denote the ranked output vector as $Y_R$.

\item[Step 7:] For each uncertain parameter, fit two multiple linear regression (MLR) models. The first one is the MLR of $X_{iR}, i=1,2,\cdots,6,$ on all $\{X_{jR}: j=1,2,\cdots,6 \ \ \text{and} \ \ j\neq i\}$. The second one is the MLR of $Y_R$ on all $\{X_{jR}: j=1,2,\cdots,6 \ \ \text{and} \ \ j\neq i\}$.

\item[Step 8:]  For each of the two fitted MLR models, calculate the residuals. For the $i$-th uncertain parameter, the PRCC is obtained by calculating the Pearson's correlation coefficient between these two sets of residuals. Compute the PRCC value for each uncertain parameter.

\item[Step 9:] For each uncertain parameter, use the student's t-test and the corresponding $p$-value to check if the PRCC value is significantly different from zero.

\item[Step 10:] Identify the sensitive parameters, i.e., the parameters having large PRCC values (e.g., $>0.5\ \text{or}\ <-0.5$) and small $p$-values (e.g., $<0.01, <0.05,\ \ \text{or}\ < 0.10$ ). These are the parameters that significantly affect the tumor volume in a colon cancer patient. A PRCC value with a positive (negative) sign implies that the corresponding parameter is directly (inversely) related to the output measure.

\item[Step 11:] Rank the identified sensitive parameters based on the magnitude of their PRCC values.
\end{enumerate}
\end{algorithm}

\section{Numerical results} \label{sec:results}

In this section, we present the results of numerical simulations that validate the effectiveness of the FP framework. We first consider the parameter estimation problem given in \eqref{eq:min_problem}.  For this purpose, we choose our domain $\Omega = (0,6)^3$ and discretize it using $N_{x_i} = 51$ points for $i = 1,2,3$. The final time $t$ is chosen to be 4 and the maximum number of time steps $N_t$ is chosen to be 50. The patient data is represented by the target PDFs $f^*_i(x),~ i = 1,\cdots,N$ with $N=10,20$, where $f^*_i$ are described by a normal distribution about the measured mean value $\mathbb{E}[f^*_i]$ and variance 0.05. We perform a 4D interpolation to obtain the data function $f^*(x,t)$ at all discrete times $t_k,~ k = 1,\cdots, N_t$. The regularization parameters are chosen to be $\alpha = 1,~ \beta = 0.02$. For the set $U_{ad}$, the value of the vector $M = (M_1,\cdots,M_6)$ is given as $(1.5,0.05,0.2,1.5,0.5,1)$. This is motivated by the maximum of the biological range of values of the parameters provided in \cite{Cser2019}. For the parameter estimation process, the initial guess of the parameter set $\bT$ in the NCG algorithm is given by $\bT_0 = (0.1,0.1,0.1,0.1,0.1,0.1)$.

For the values of $\sigma_i$, we analyze the data given in \cite{Cser2019}, corresponding to the measurement of $V$ in 10 mice on days 3, 5, 7, 9, 11, 13, 15, 17, 19. We first compute the mean of $V$ for each of the corresponding days in the sample. We remark that in past, the standard deviation $\sigma$ of the data corresponding to a dynamical variable $D_y$ has been given by the form $mD_y$. While being a reasonable approximation, this expression is not accurate and moreover, such an expression will lead to degenerate elliptic coefficients in the FP equation \eqref{eq:FP}. Thus, to overcome these issues, we use a data fitting method to come up with a more accurate form of $\sigma$ such that it is not degenerate. We assume that standard deviation of the aforementioned dataset is given by the form $ m(\bar{V}^d)^{ra}+\eps$, where $\eps > 0$ and $\bar{V}^d$ is the mean of $\bar{V}$ on day $d$. We find that for the choice of $ra=1.2,~ m = 0.5,\eps=0.001$, the standard deviation of the data is well-fitted by $m(\bar{V}^d)^{ra}+\eps$. Thus, in \eqref{eq:FP}, we choose
\[
\sigma_i(x) = 0.5(x_i^{1.2}+0.001),~ i = 1,2,3.
\]

\subsection{Test Case 1: Synthetic Data} In this test case, we generate synthetic data measurement, that represents a hypothetical colon cancer patient, by solving the ODE \eqref{eq:NODE} in the time interval $t = [0,4]$ with $N = 10,20$, and with the non-dimensional parameter set $\bT = (1.3400, 0.0350, 0.1200, 1.1400, 0.2473, 0.5000)$. The values of the constants used in converting the ODE system \eqref{eq:ODE} to its non-dimensional form given in \eqref{eq:NODE} are given as $k_1 = \frac{1}{10},~k_2 = 3,~k_3 = 100,~k_4 = \frac{1}{10}$. The data is given as $(\bar{V}_i,\bar{B}_i, \bar{T}_i) = (\bar{V}(t_i),\bar{B}(t_i), \bar{T}(t_i)) ,~i = 1,\cdots,N$, for specified times $t_i$, with $(\bar{V}(0),\bar{B}(0), \bar{T}(0) = (1,1,1))$. The initial condition represents a cancer-free state of the patient. The corresponding PDFs $f^*_i$ are given by normal distribution functions with mean $\bar{V}_i$ and variance 0.05. 

\begin{figure}[H]
\centering
\subfloat[$V$ - MC plots]{\includegraphics[width=0.35\textwidth, height=0.30\textwidth]{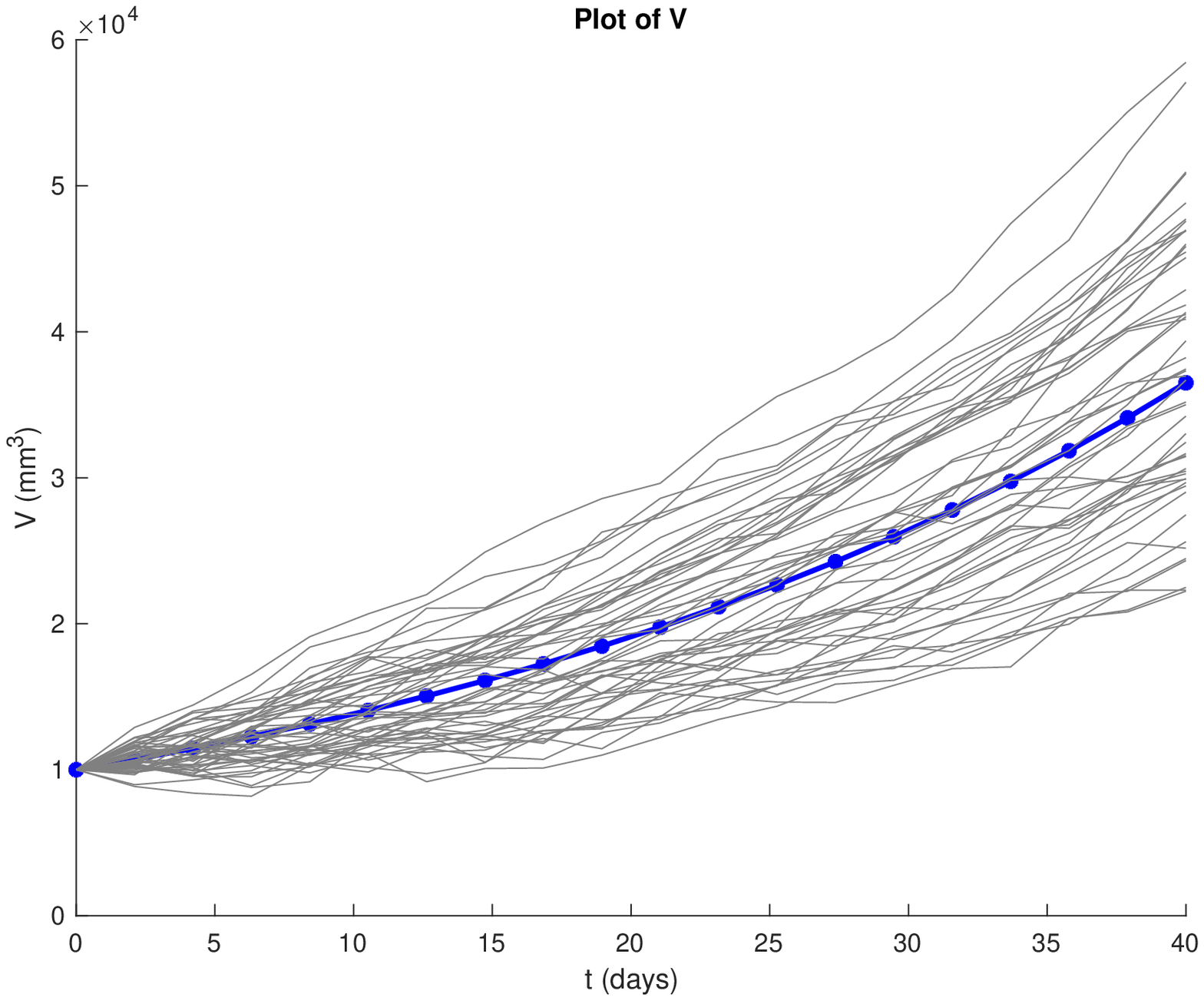}\label{V10_MC}}
\subfloat[$B$ - MC plots]{\includegraphics[width=0.35\textwidth, height=0.30\textwidth]{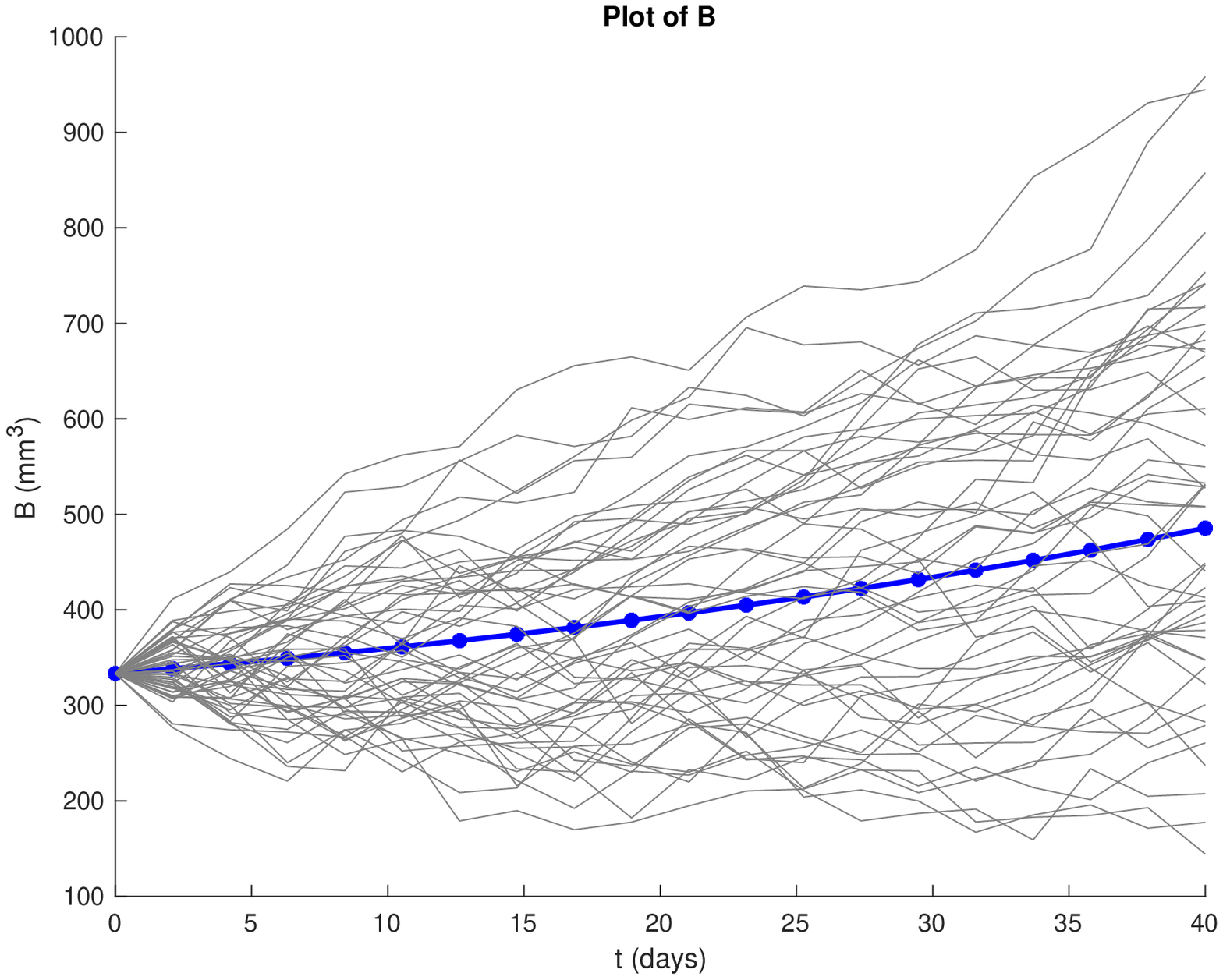}\label{B10_MC}}
\subfloat[$T $- MC plots]{\includegraphics[width=0.35\textwidth, height=0.30\textwidth]{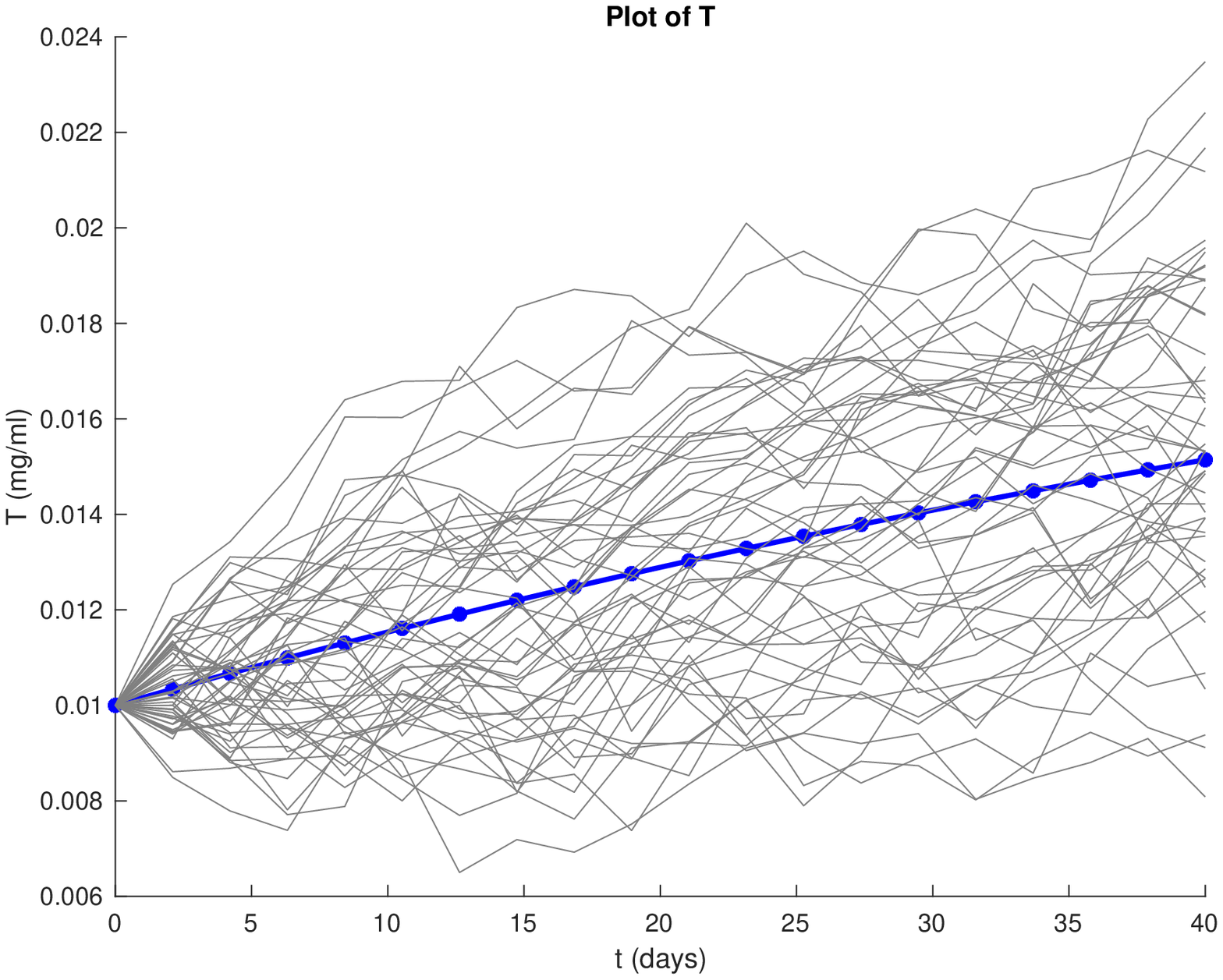}\label{T10_MC}}\\
\subfloat[$V$]{\includegraphics[width=0.35\textwidth, height=0.30\textwidth]{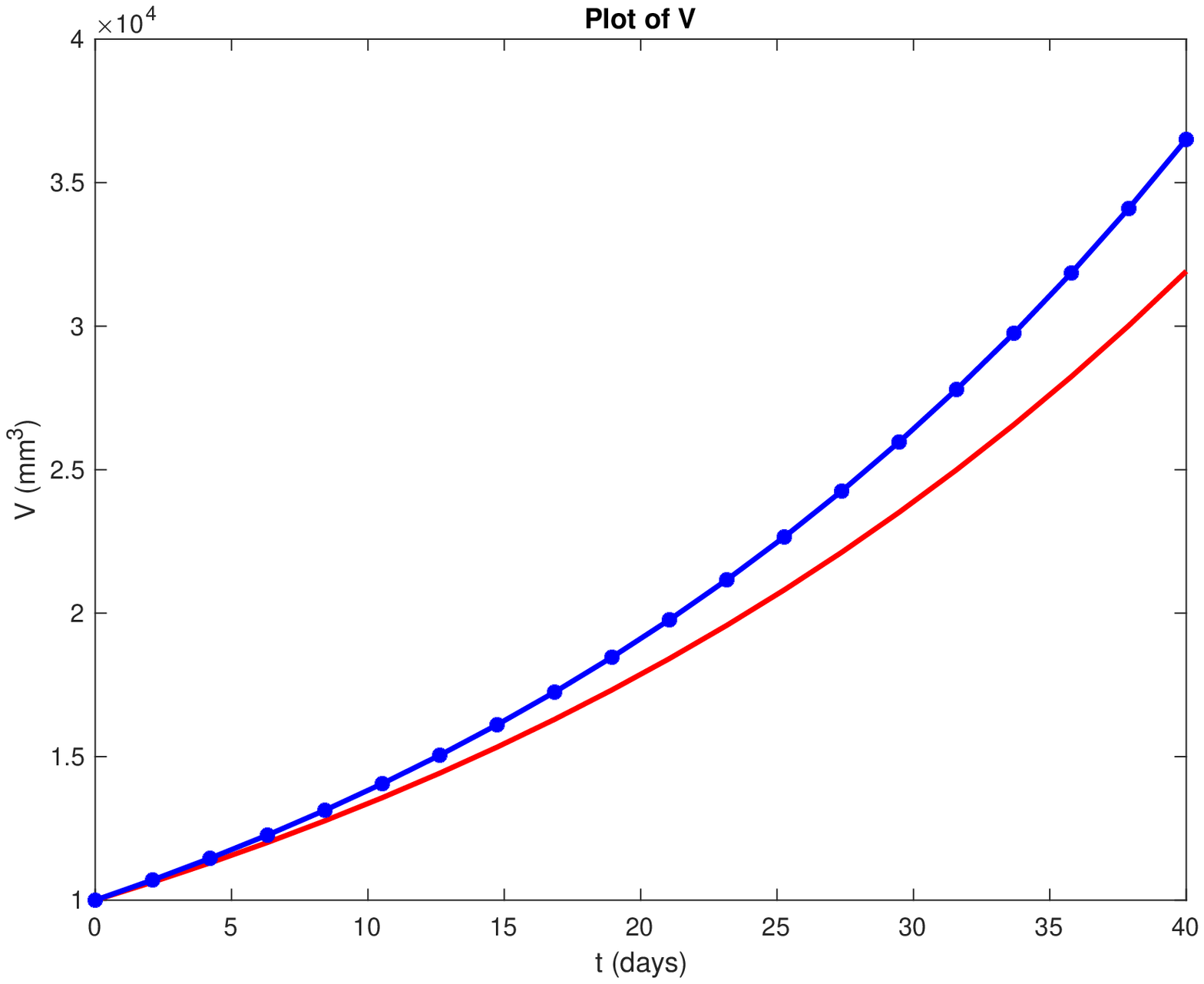}\label{V10}}
\subfloat[$B$]{\includegraphics[width=0.35\textwidth, height=0.30\textwidth]{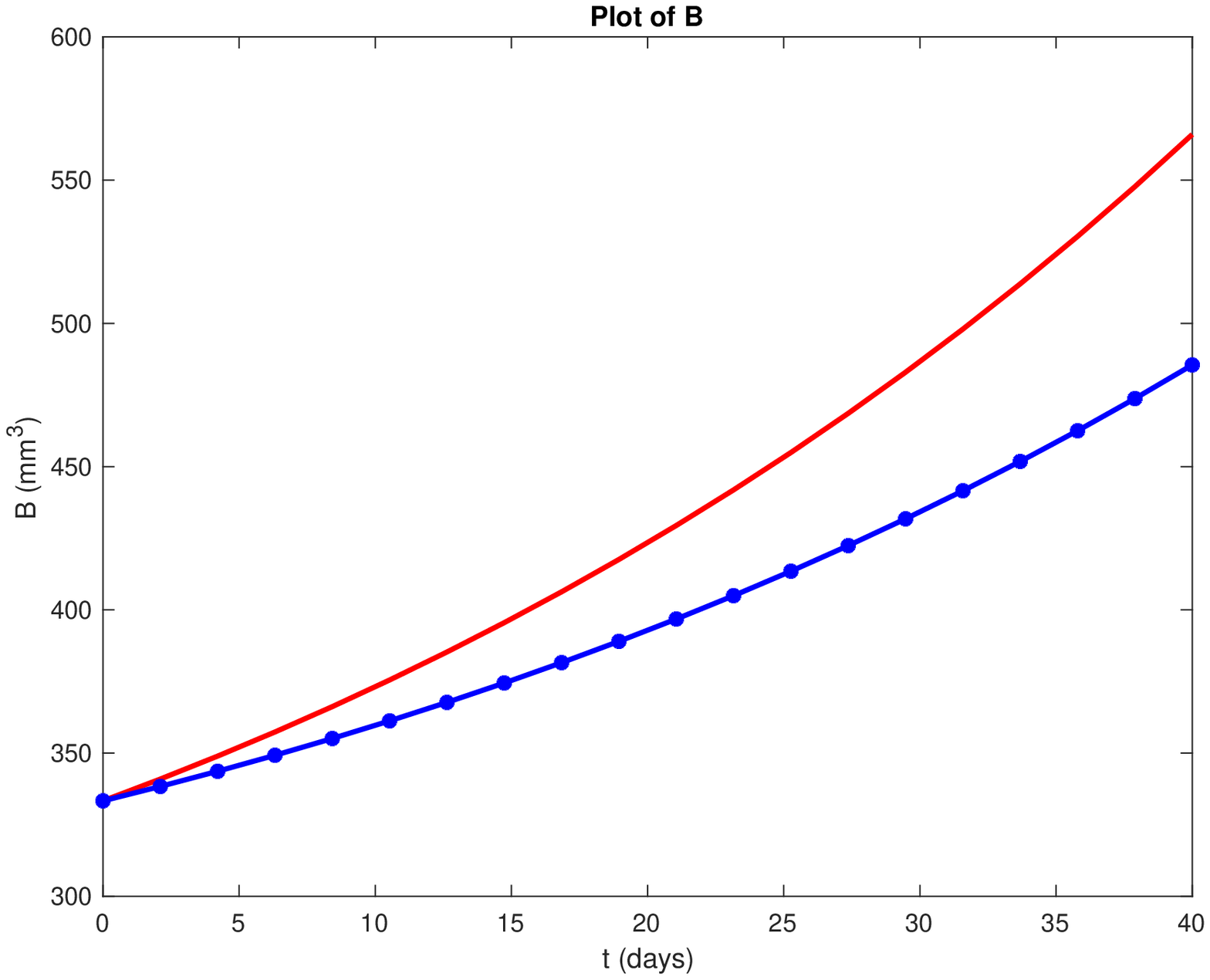}\label{B10}}
\subfloat[$T$]{\includegraphics[width=0.35\textwidth, height=0.30\textwidth]{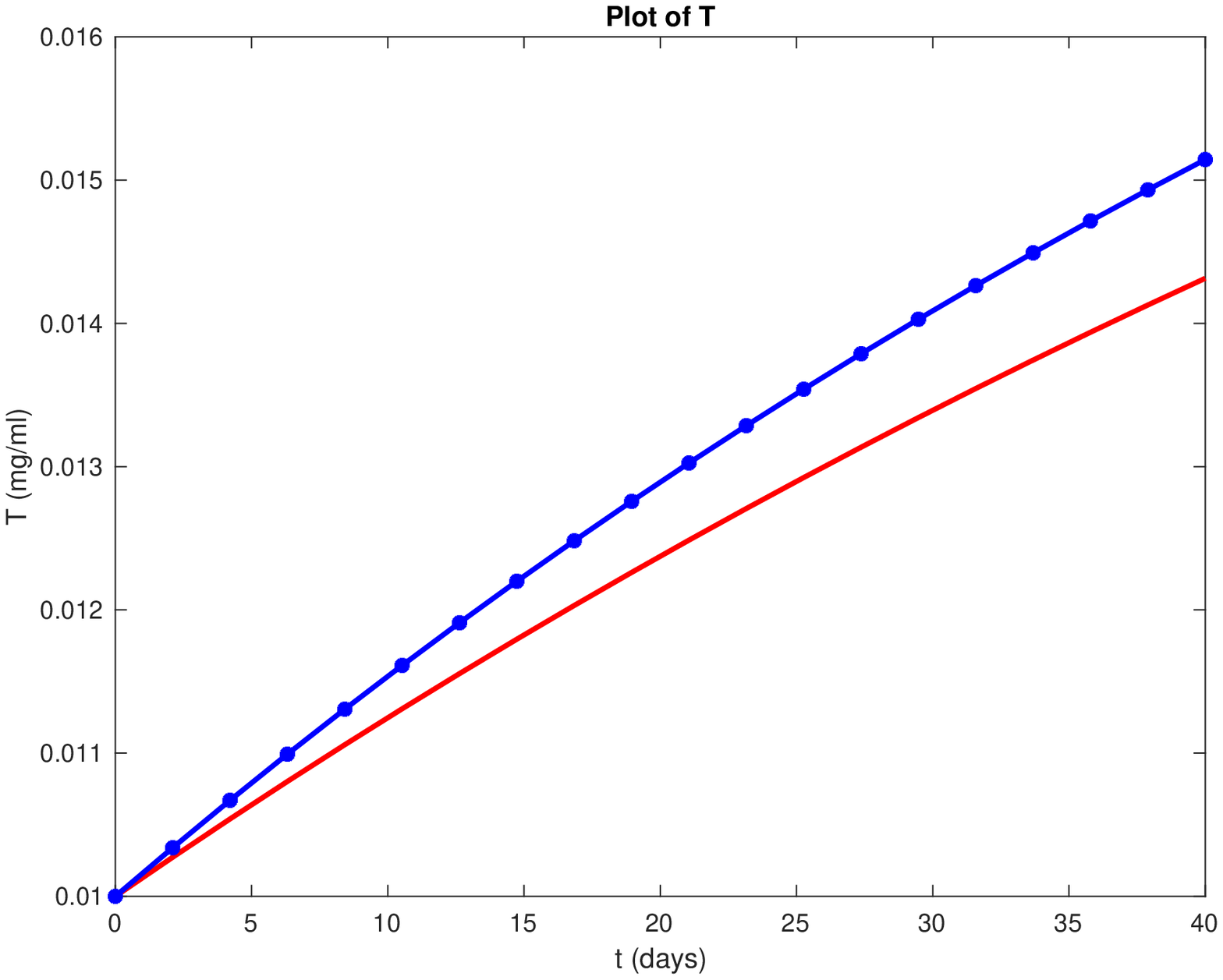}\label{T10}}
 \caption{Test Case 1: Monte-Carlo simulation and mean trajectory plots of the profiles of $V,B,T$ with $N=10$. }
    \label{fig:test_case1}
  \end{figure}
  
 The obtained optimal parameter set $\bT^*$ for $N=10$ is  $(0.9349,    0.042,    0.1478,    0.8505, 0.18,    0.6726)$ and with $N=20$ is $( 1.0359  , 0.0383,    0.1940,    1.0877,    0.2000,   0.5775)$. We observe that the optimal parameter estimates get closer to the true parameter values with the increase in $N$, which is expected due to availability of additional data. Figure \ref{fig:test_case1} shows the 50 Monte-Carlo (MC) simulation plots and profiles of $V,B,T$ with $N=10$. The plots in the first row show the Monte-Carlo simulations for the solution of the stochastic ODE \eqref{eq:ItoODE} with the true parameter set. The plots in the second row show the mean trajectories of $V,B,T$ with the true and optimal parameter set. The blue curve in each of the figures show the plot of the true mean value of the corresponding random variables $(V,B,T)$. The red curve shows the plot of the mean value of the corresponding variables obtained by solving the \eqref{eq:ItoODE} with the optimal parameter set. The Monte-Carlo simulation shows a large variance in the data due to measurement errors. Using such a data and the FP framework, the optimal parameter set $\bT^*$ leads to the mean value of the variables being close to the true mean value. This demonstrates that even in the presence of significant variations in the measurement data, the obtained optimal parameter set is close to the true values and provides accurate mean values of the random variables $(V,B,T)$.

\begin{figure}[H]
\centering
\subfloat[$V$ - MC plots]{\includegraphics[width=0.35\textwidth, height=0.30\textwidth]{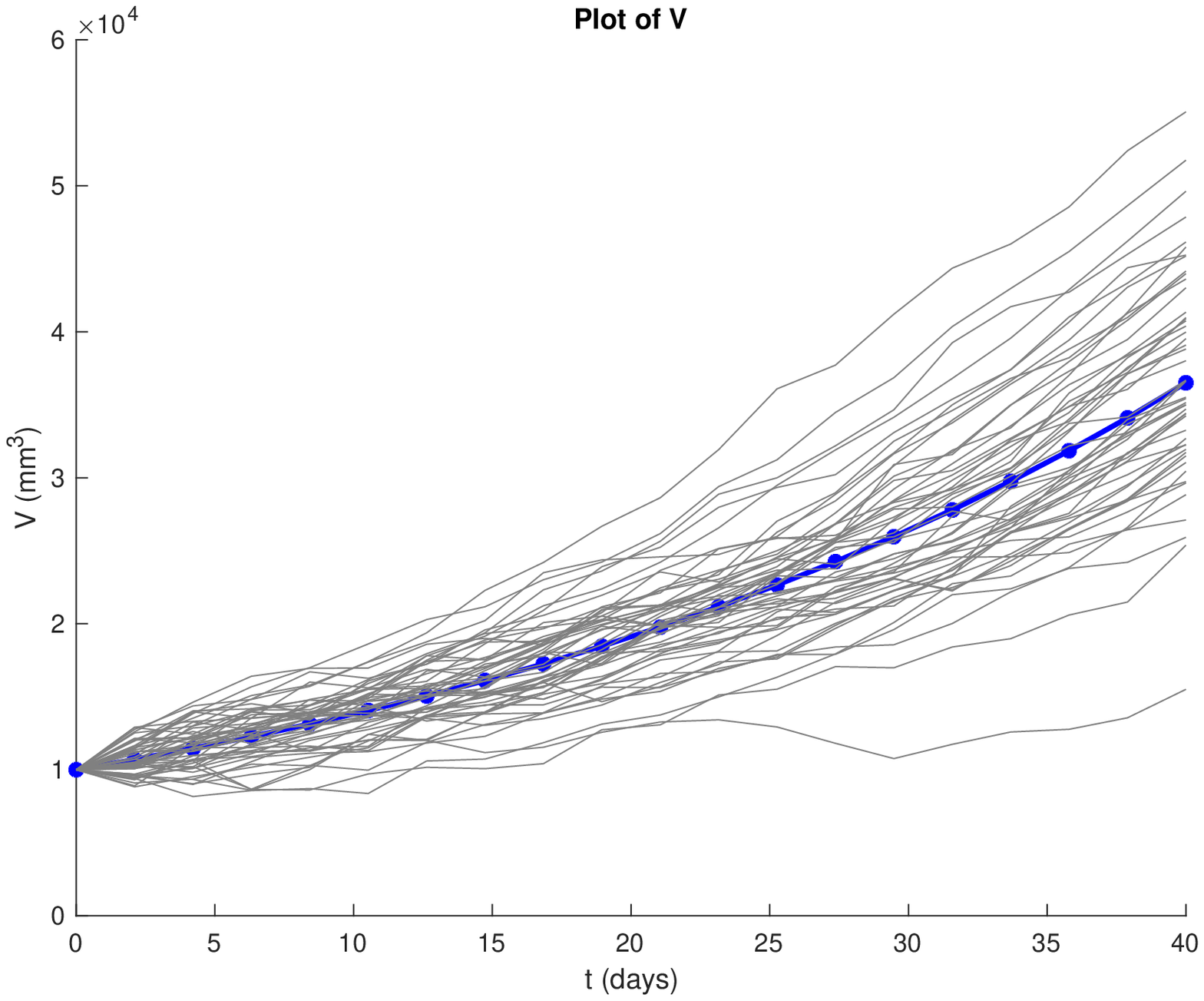}\label{V20_MC}}
\subfloat[$B$ - MC plots]{\includegraphics[width=0.35\textwidth, height=0.30\textwidth]{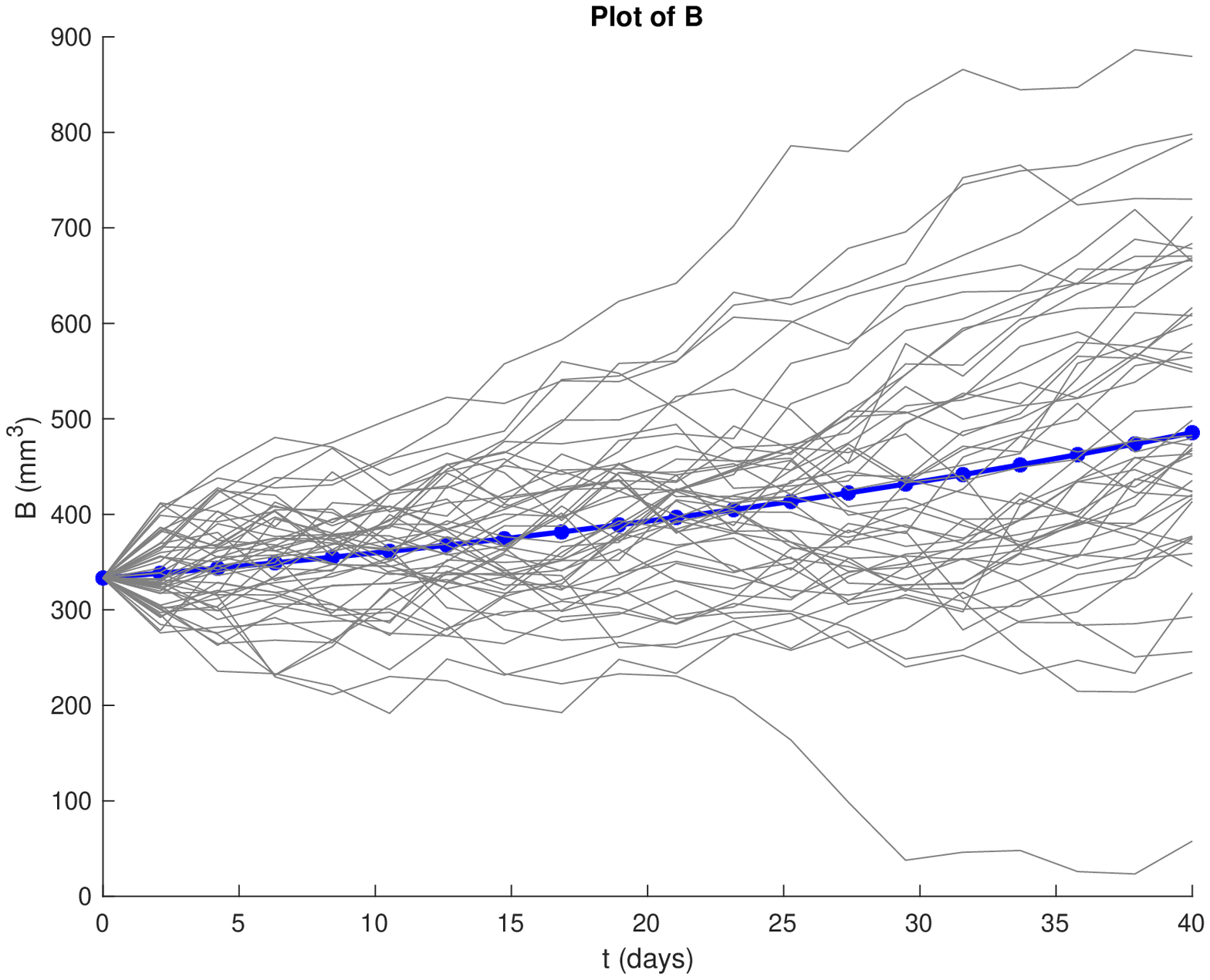}\label{B20_MC}}
\subfloat[$T $- MC plots]{\includegraphics[width=0.35\textwidth, height=0.30\textwidth]{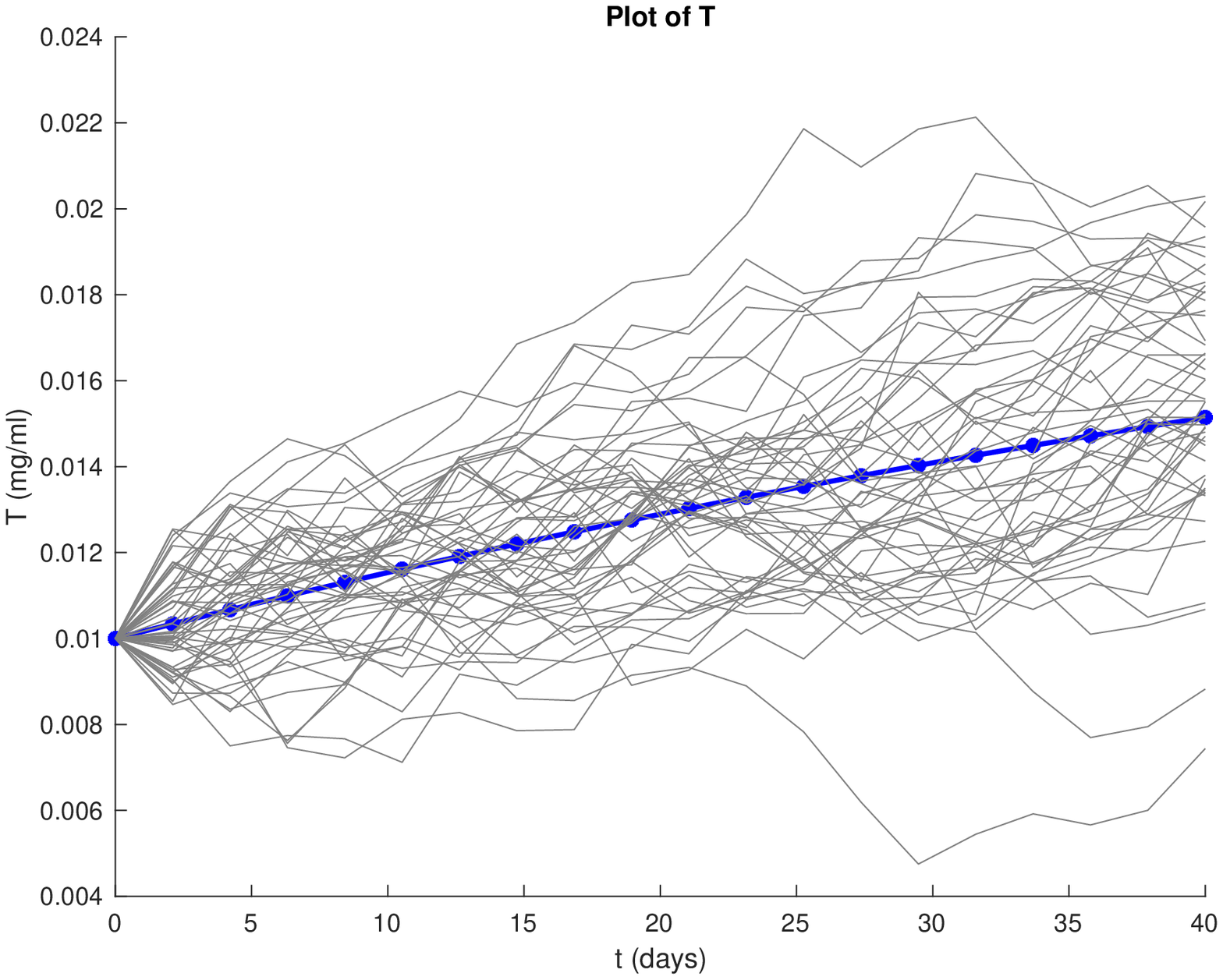}\label{T20_MC}}\\
\subfloat[$V$]{\includegraphics[width=0.35\textwidth, height=0.30\textwidth]{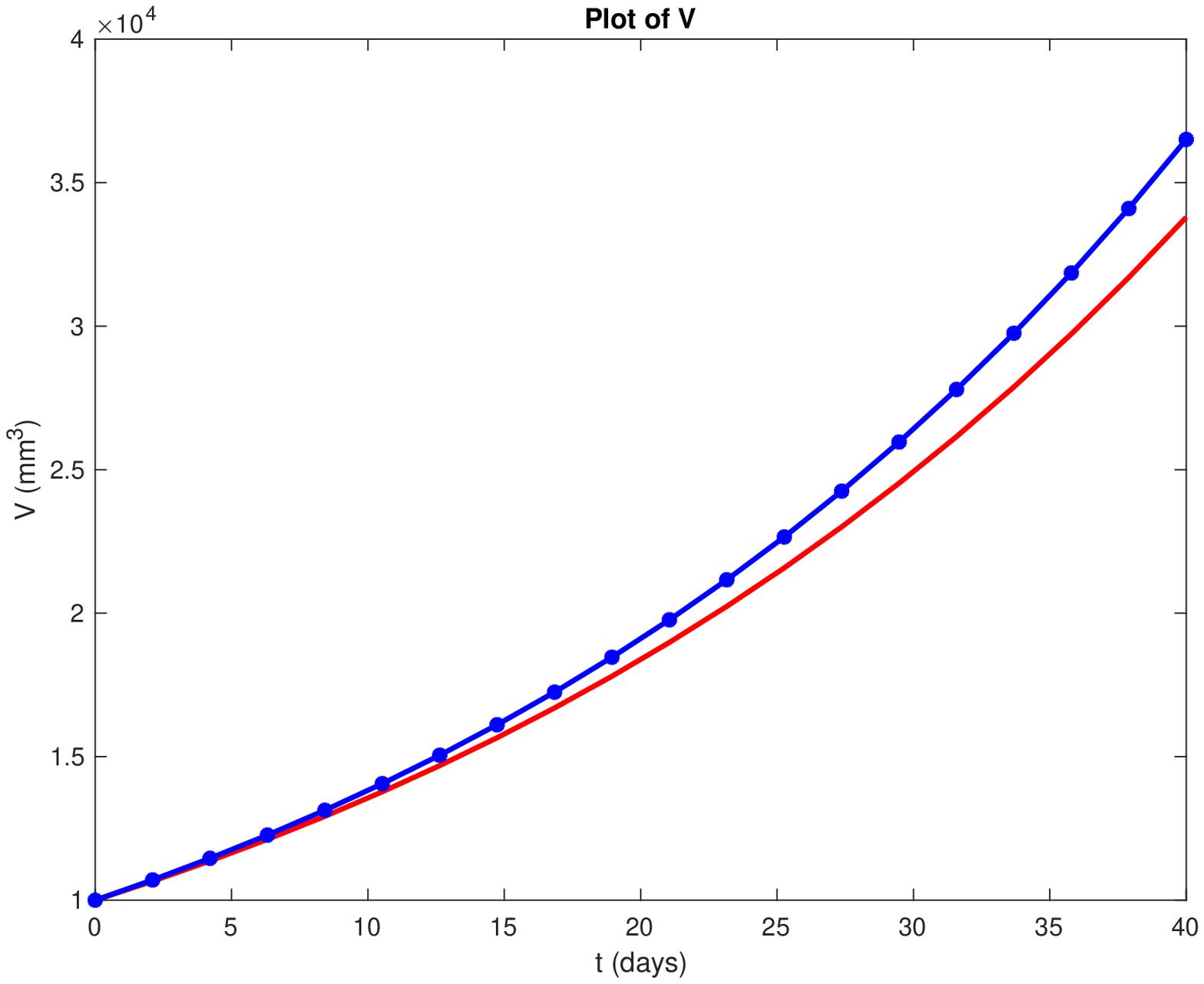}\label{V20}}
\subfloat[$B$]{\includegraphics[width=0.35\textwidth, height=0.30\textwidth]{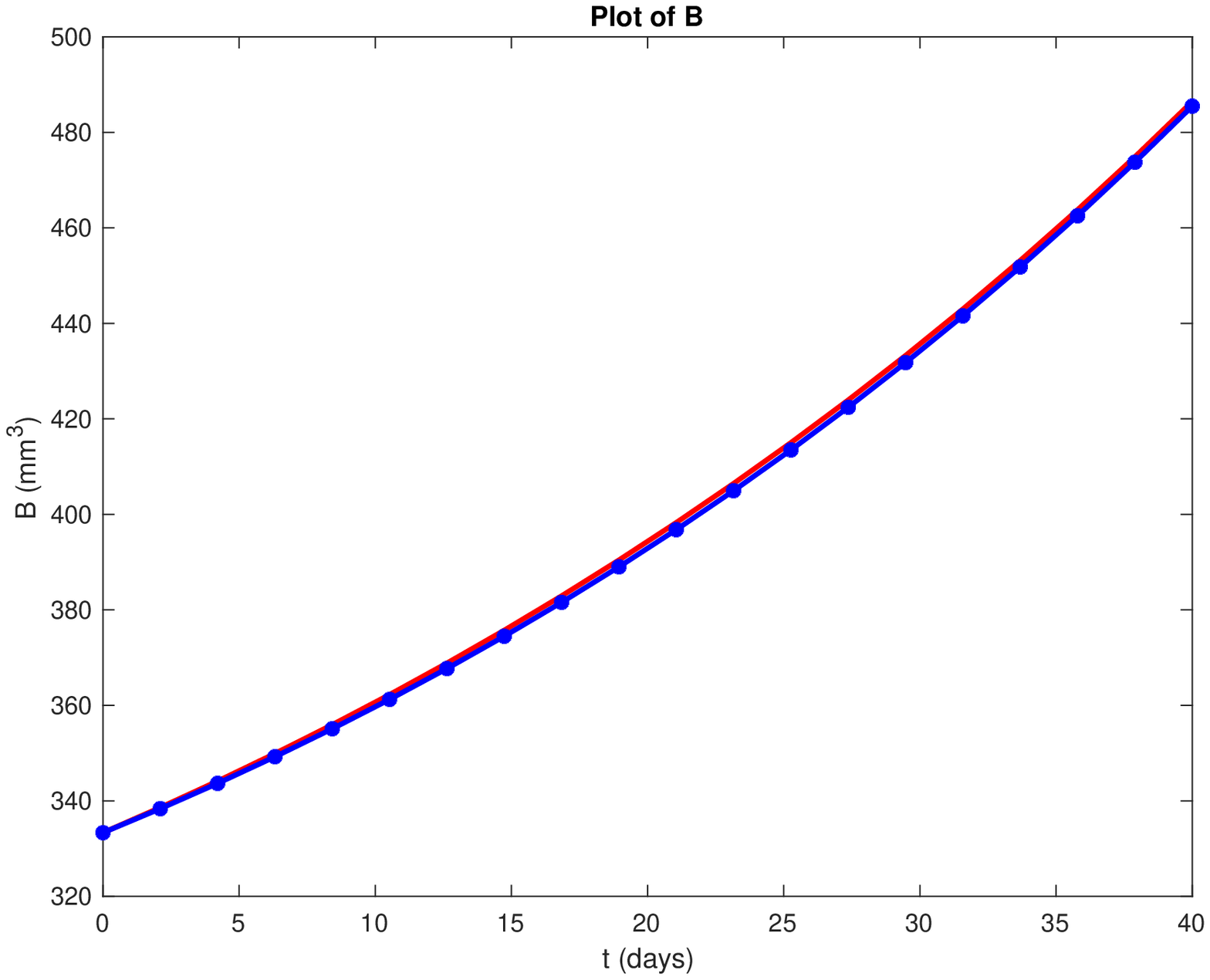}\label{B20}}
\subfloat[$T$]{\includegraphics[width=0.35\textwidth, height=0.30\textwidth]{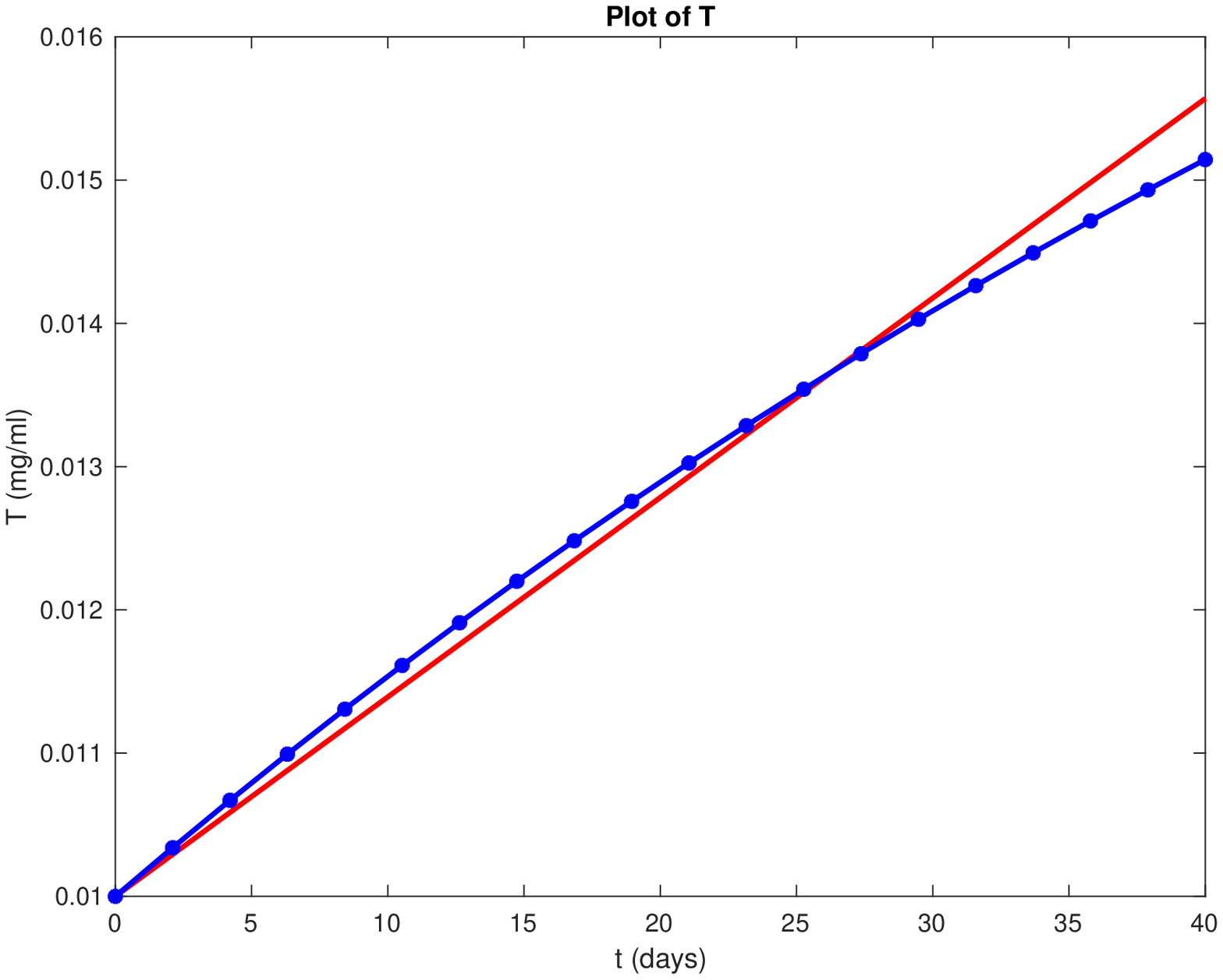}\label{T20}}
 \caption{Test Case 1: Monte-Carlo simulation and mean trajectory plots of the profiles of $V,B,T$ with $N=20$. }
    \label{fig:test_case1_20}
  \end{figure}
  Furthermore, in Figure \ref{fig:test_case1_20}, we have the simulations for $N=20$. We observe that with an increase in the number of data points, the obtained mean value of the random variables get closer to the true mean value.
  
Next, we run the LHS-PRCC algorithm and investigate the sensitivity of the optimal parameter set with respect to the tumor volume $V$ at the final time $T=4$. For this purpose, we assume each parameter to follow a Weibull distribution \cite{Pal2020,Pal2021} and consider the number of equiprobable intervals, $M$, to be 100. The $p-$values and the PRCC values for the cases $N=10$ and $N=20$ are given in Table \ref{tab:pval1}.

\begin{table}[H]
\begin{center}
\begin{tabular}{V{3}cV{3}lrV{3}lrV{3}}
\hlineB{3}
\textbf{Parameter} & $N=10$ & & $N=20$ & \\ \cline{2-3} \cline{4-5}
& $p-$value & PRCC value & $p-$value & PRCC value\\
\hlineB{3}
$c$  &       2.652e-30 &  0.870                    & 1.842e-29 &  0.864 \\
\hline
$c_e$  &   0.629 &-0.050                            & 0.913 &0.011   \\
\hline
$c_v$  &   0.302 & 0.107                            & 0.622 &-0.051   \\
\hline
$c_T$  &   0.840 & 0.021                            & 0.368 &0.093  \\
\hline
$q_T$  &   0.696 & 0.041                           & 0.125 &0.158  \\
\hline
$\gamma$  &  1.301e-44 & 0.938              & 2.933e-48 &0.949\\
\hlineB{3}
\end{tabular}
\end{center}
\caption{$p-$values and PRCC values for the optimal parameter set $\bT^*$}
\label{tab:pval1}

\end{table}
We observe that the parameters $c$ and $\gamma$ have $p-$values close to 0 and high PRCC values, which make them the most sensitive variables with respect to the tumor volume $V$. This is expected as these two parameters directly influence the rate of increase of $V$. As far as the other parameters are concerned, we note that their $p-$values are high and PRCC values are low. Hence, these parameters are not sensitive to the output $V$. Between the parameters $c$ and $\gamma$, since the PRCC value of $\gamma$ is higher than that of $c$, we can say that the parameter $\gamma$ is more sensitive to $V$ than the parameter $c$. 

We next determine the optimal treatment strategies for curing the colon cancer patient. For this purpose, we note that since $c,\gamma$ are the most sensitive parameters with respect to the tumor volume $V$, it is enough to use a combination drug that can control the effects of $c,\gamma$. So we test a combination therapy comprising of Bevacizumab and Capecitabine that directly affects $T$ and $V$, and so we don't consider $u_2$ in \eqref{eq:ODE} by setting $\alpha_2 = 0$. This is because the joint effect of $c\gamma$ is similar for both $V$ and $B$ of \eqref{eq:ODE}, and it is enough to consider a drug that controls either $V$ or $B$ directly. 

We start off with the values of $(\bar{V},\bar{B},\bar{T})$, as obtained in Figure \ref{fig:test_case1_20} at the final non-dimensionalized time $t=4$. This stage is represented by the presence of colon cancer in the patient. Our target is to use the optimal combination therapy to bring down the levels of $(\bar{V},\bar{B},\bar{T})$ to (1,1,1) (that represents a cancer-free stage) after dosage administration for a subsequent non-dimensionalized time period of $t\in[0,1.4]$, which corresponds to the actual time of $\tau=14$ days. This time period corresponds to the standard combination treatment cycle with Bevacizumab and a chemotherapy drug for colon cancer \cite{Chong2010}. Thus, we consider $f_d(x)$ to be a Gaussian centered about $(1,1,1)$ at time $t=1.4$. We next solve the minimization problem \eqref{eq:second_min_problem} to obtain the optimal profiles for $u_1$ and $u_3$, representing dosages of Capecitabine and Bevacizumab, respectively. The values of $\alpha_1,\alpha_3$ are chosen to be $2.135 \cdot 10^{-6}$ \cite{Frances2011} and 0.16 \cite{Cser2019}, respectively. The maximum tolerable dosage for Capecitabine is taken to be $1000$ mg/m$^2$/day and for Bevacizumab is taken to be $0.7$ mg/kg/day, which implies, $D_1 = 0.02135$ and $D_3 = 1.12$.

\begin{figure}[H]
\centering
\subfloat[$V$]{\includegraphics[width=0.35\textwidth, height=0.30\textwidth]{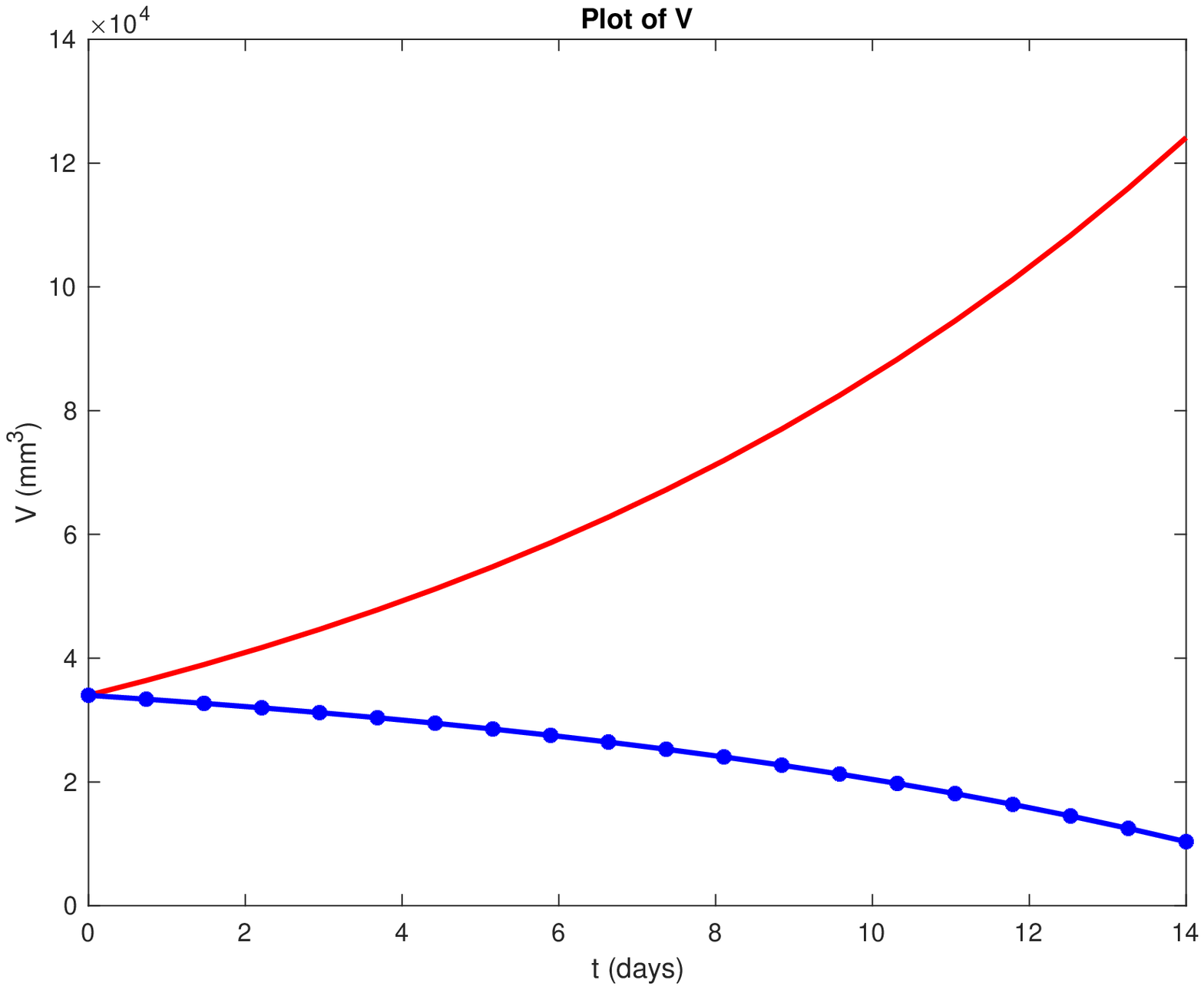}\label{V20_t}}
\subfloat[$B$]{\includegraphics[width=0.35\textwidth, height=0.30\textwidth]{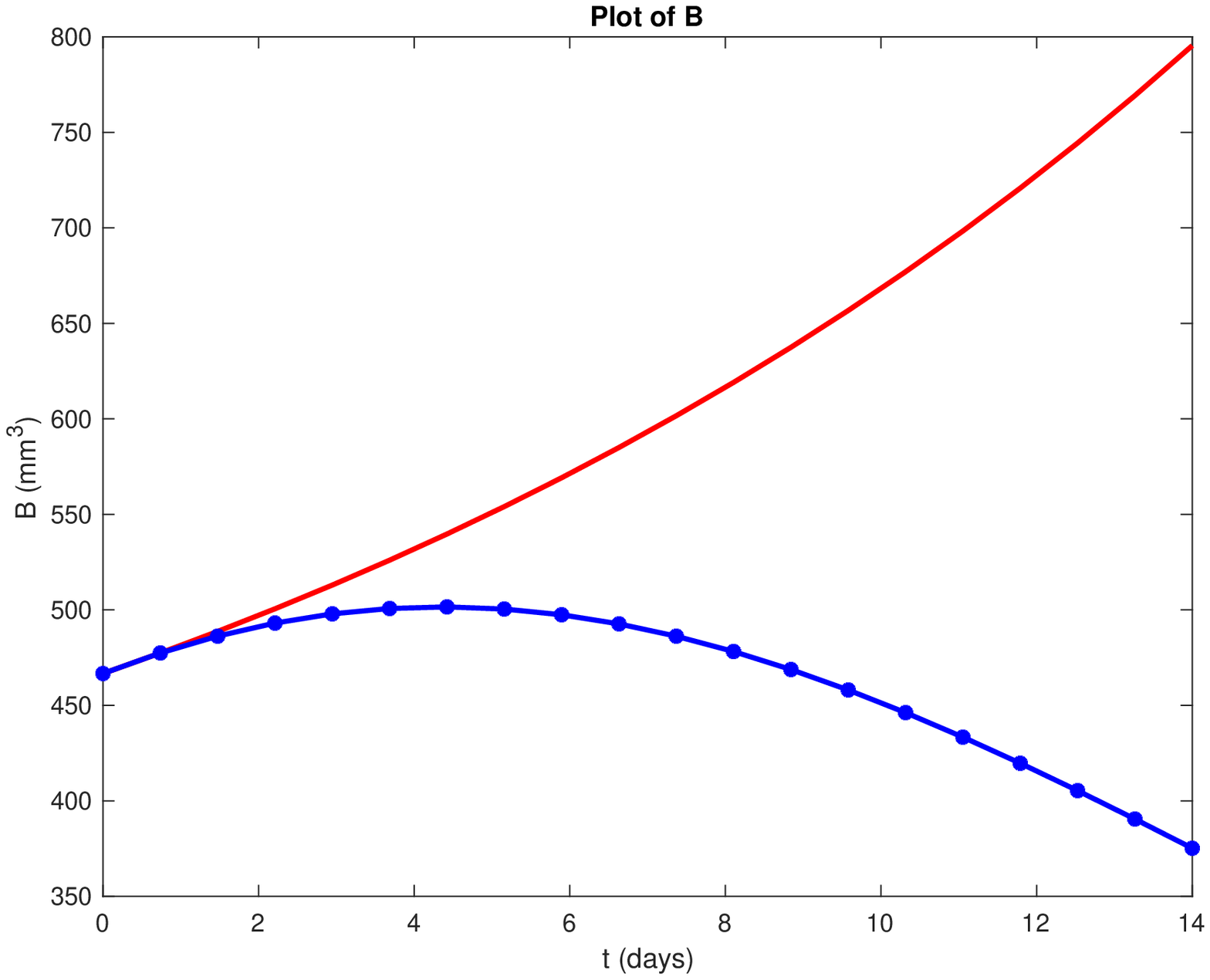}\label{B20_t}}
\subfloat[$T$]{\includegraphics[width=0.35\textwidth, height=0.30\textwidth]{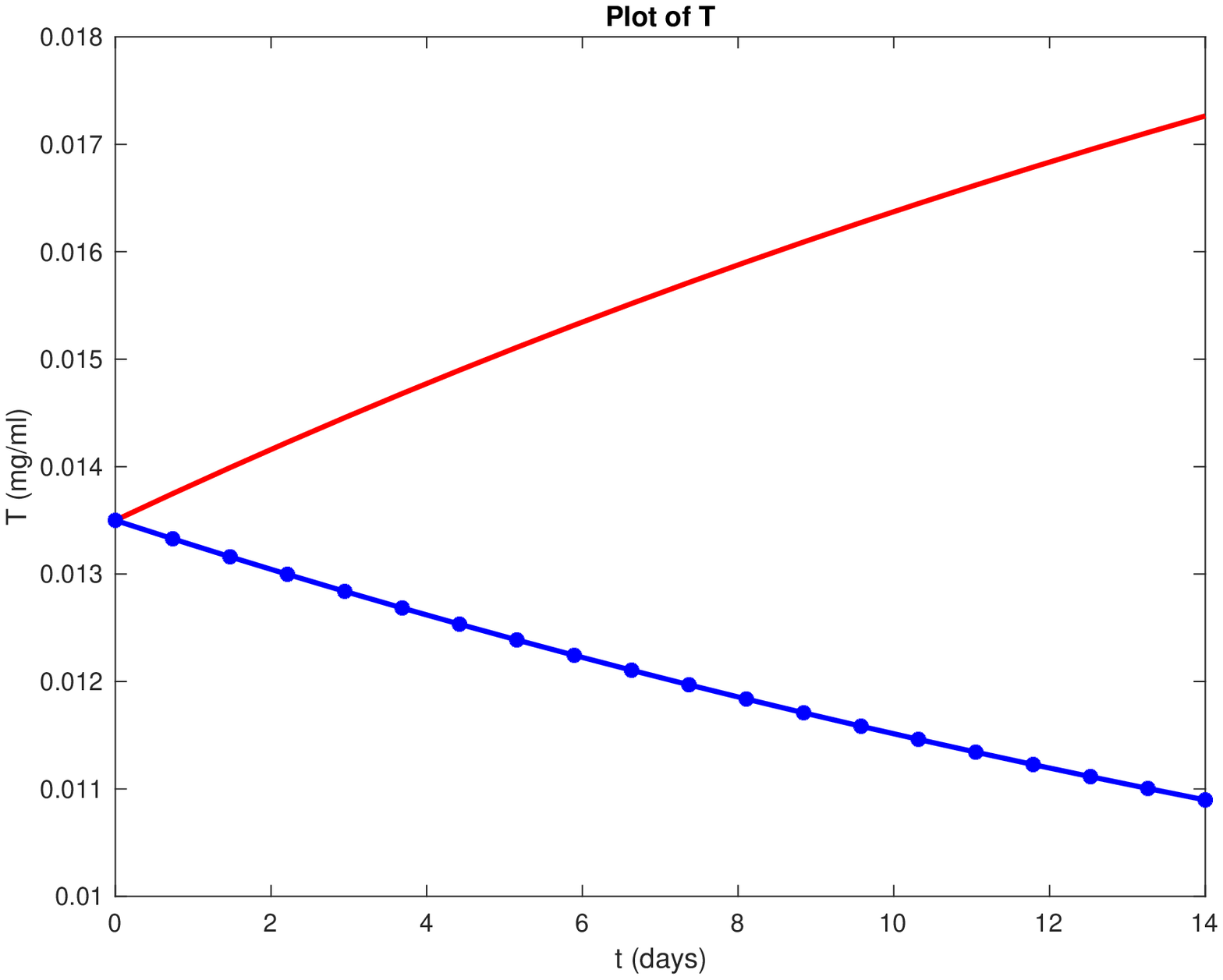}\label{T20_2}}
 \caption{Test Case 1: Mean trajectory plots of the profiles of $V,B,T$ with and without the combination therapy. Red curves indicate the profiles without treatment, blue curves indicate the profiles with treatment }
    \label{fig:test_case1_treat}
  \end{figure}

Figure \ref{fig:test_case1_treat} shows the plots of the mean PDFs with and without treatment. The red curves denotes profiles of $V,B,T$ without treatment, whereas the blue curves denotes profiles with the effect of treatments. We clearly note that in the absence of the combination therapy, the values of $V,B,T$ clearly rise uncontrolled, which means that the tumor is rapidly spreading. On the administration of the combination therapy, we observe the control of $V,B,T$ to the desired cancer-free state. Figure \ref{fig:test_case1_u} shows the dosage patterns of Capecitabine and Bevacizumab. We note that initially, Capecitabine is administered with higher dosages but over time the dosage is significantly reduced. On the other hand, the dosage of Bevacizumab is low during the initial phases of the treatment and subsequently increases with the decrease of Capecitabine.

\begin{figure}[H]
\centering
\subfloat[$u1$: Capecitabine]{\includegraphics[width=0.35\textwidth, height=0.30\textwidth]{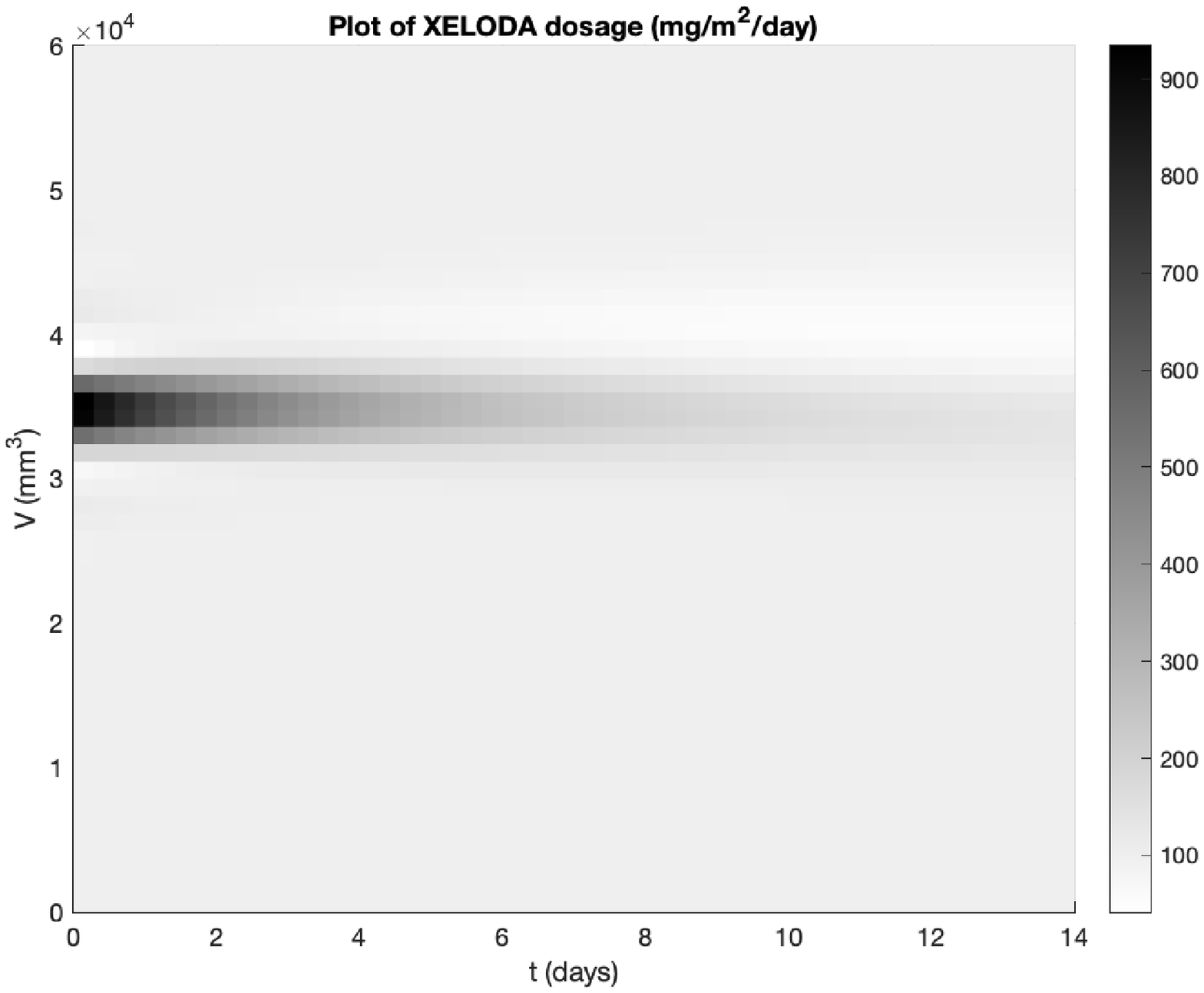}\label{u1}}
\subfloat[$u3$: Bevacizumab]{\includegraphics[width=0.35\textwidth, height=0.30\textwidth]{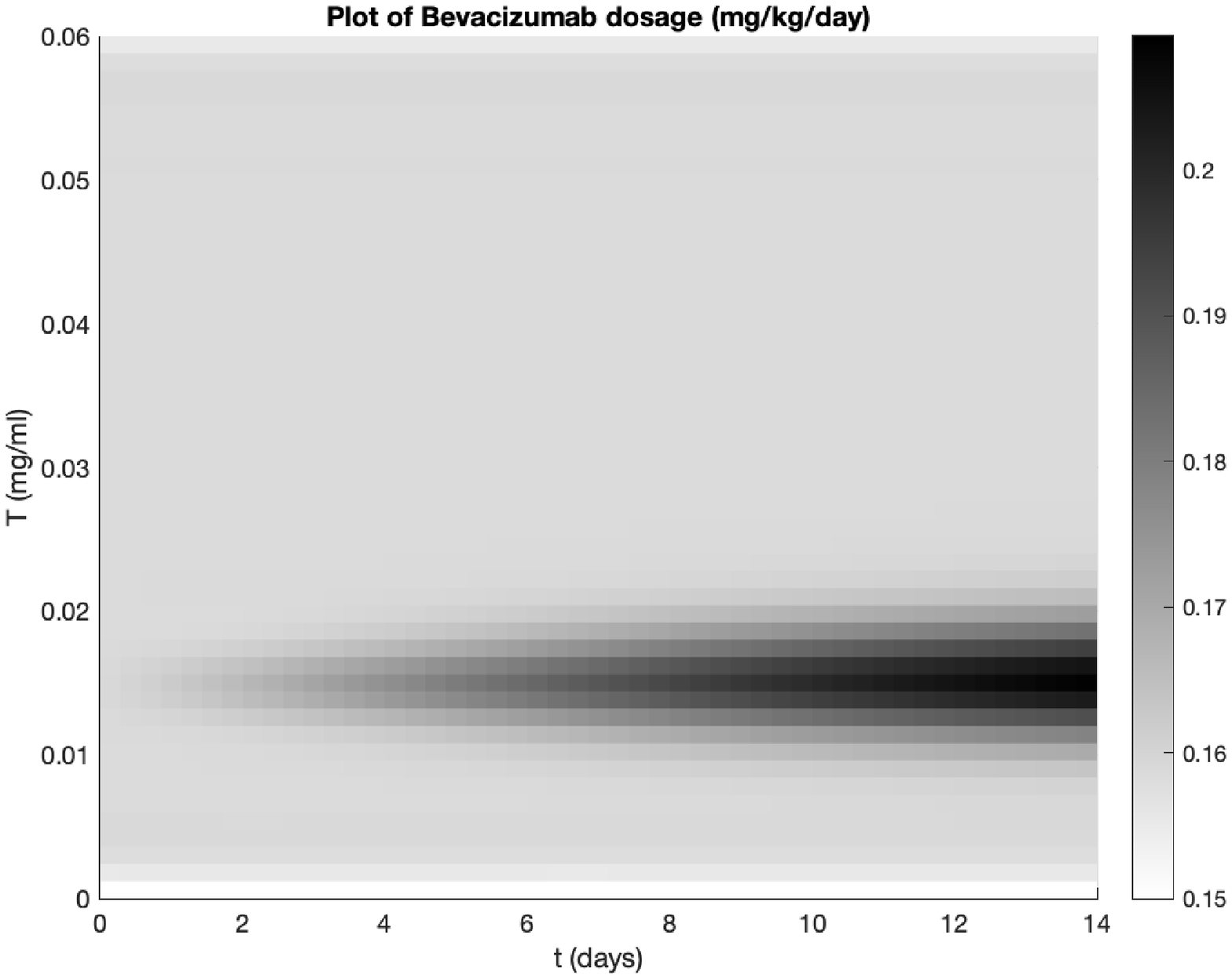}\label{u2}}
 \caption{Test Case 1: Feedback optimal combination treatment profiles }
    \label{fig:test_case1_u}
  \end{figure}

\subsection{Test Case 2: Real data}
 In Test Case 2, we use real data based on experiments in \cite{Cser2019,Sapi2015}. In \cite{Sapi2015}, mice specimens were transplanted subcutaneously with C38 colon adenocarcinoma, and small animal MRI was used to measure the tumor volume $\bar{V}$ in days 3, 5, 7, 9, 11, 15, 17, 19, 21, 23. The corresponding data for $\bar{B}$ and $\bar{T}$ were taken from \cite{Cser2019}. With each of the the data points $(\bar{V}_i,\bar{B}_i,\bar{T}_i)$ as the mean of PDFs $f_i^*$, with variance 0.05, the mathematical data function $f^*(x,t)$ was generated using 4D interpolation. For implementing our NCG algorithm, we used the time interval $t = [0,2.3]$. The values of the constants used in converting the ODE system \eqref{eq:ODE} to its non-dimensional form given in \eqref{eq:NODE} are given as $k_1 = \frac{1}{4},~k_2 = 50,~k_3 = 100,~k_4 = \frac{1}{10}$.

\begin{figure}[H]
\centering
\subfloat[$V$]{\includegraphics[width=0.35\textwidth, height=0.30\textwidth]{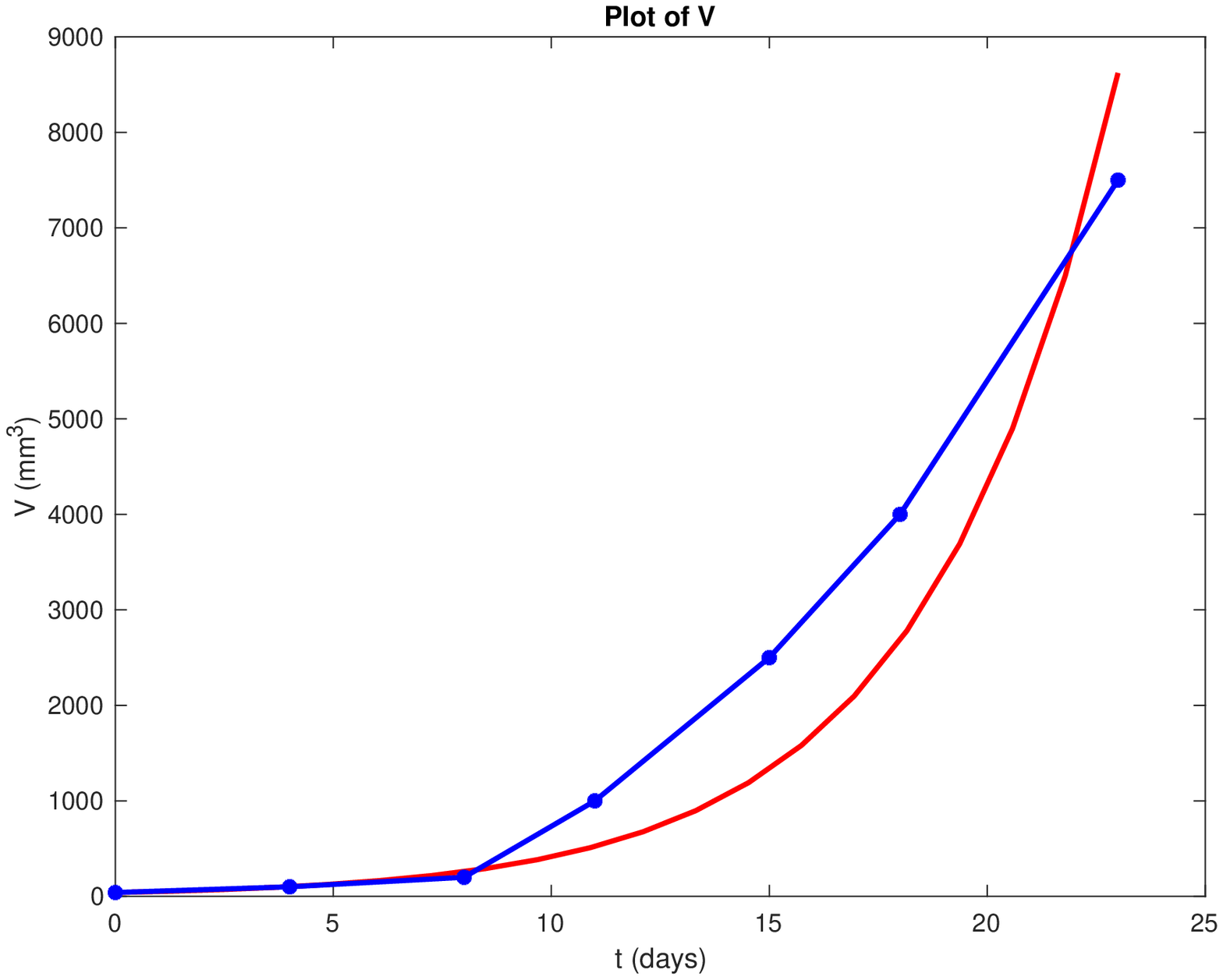}\label{Vreal}}
\subfloat[$B$]{\includegraphics[width=0.35\textwidth, height=0.30\textwidth]{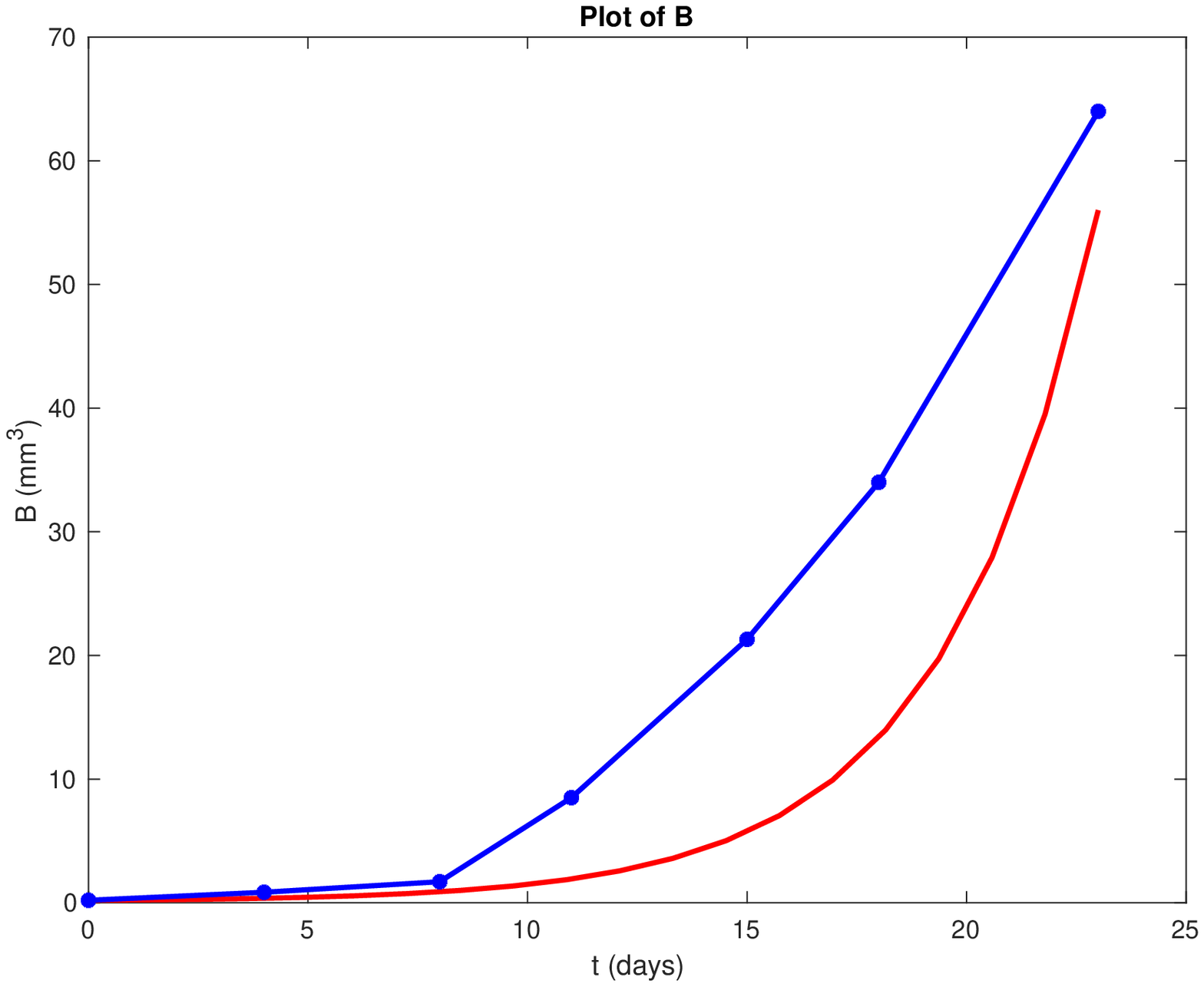}\label{Breal}}
\subfloat[$T$]{\includegraphics[width=0.35\textwidth, height=0.30\textwidth]{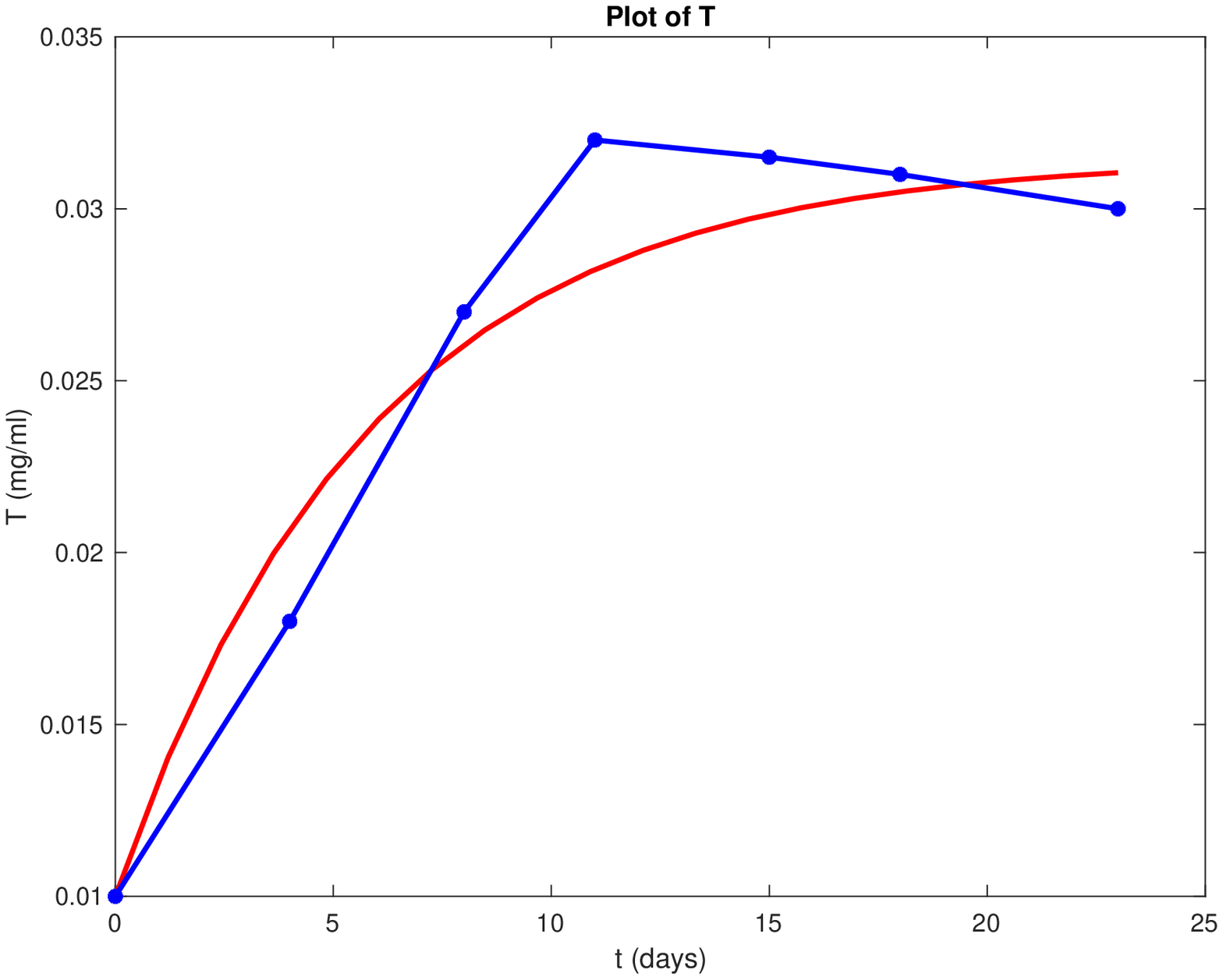}\label{Treal}}
 \caption{Test Case 2: Mean trajectory plots of the profiles of $V,B,T$ with 10 real data set points. }
    \label{fig:test_case2}
  \end{figure}
 
 The obtained value of the optimal parameter set is $\bT^* = (1.5000,   0.067,    0.2000,    1.2227,    0.3855,    0.6000)$. Figure \ref{fig:test_case2} shows the plots of the mean value of the variables $V,B,T$ using the optimal parameter set (red curve) and the true dataset points linearly interpolated (blue curve). We again observe that the obtained optimal parameter values lead to a good fit of the mean variable values to the dataset. The results of the LHS-PRCC analysis, based on Weibull density for each parameter and number of equiprobable intervals to be 100, are presented in Table \ref{tab:pval2}. From the $p-$values, it is clear that the parameters $c$ and $\gamma$ are highly sensitive to the tumor volume $V$ at the final time. From the PRCC values, we can conclude that the parameter $\gamma$ is more sensitive to $V$ than the parameter $c$.
 
 \begin{table}[H]
 
\begin{center}
\begin{tabular}{V{3}cV{3}cV{3}cV{3}}
\hlineB{3}
\textbf{Parameter} &$p-$value & PRCC value  \\
\hlineB{3}
$c$  & 7.166e-33 &  0.886\\
\hline
$c_e$  & 0.557 &-0.061  \\
\hline
$c_v$  & 0.287& 0.110  \\
\hline
$c_T$  &  0.179& 0.139 \\
\hline
$q_T$  &  0.311& -0.105 \\
\hline
$\gamma$  &  3.040e-66 & 0.979 \\
\hlineB{3}
\end{tabular}
\end{center}
\caption{$p-$values and PRCC values for the optimal parameter set $\bT^*$ corresponding to the real data}
\label{tab:pval2}

\end{table}

Since $c,\gamma$ are again determined to be the most sensitive parameters with respect to $V$, we again determine optimal combination therapies using Capecitabine and Bevacizumab as in Test Case 1.

\begin{figure}[H]
\centering
\subfloat[$V$]{\includegraphics[width=0.35\textwidth, height=0.30\textwidth]{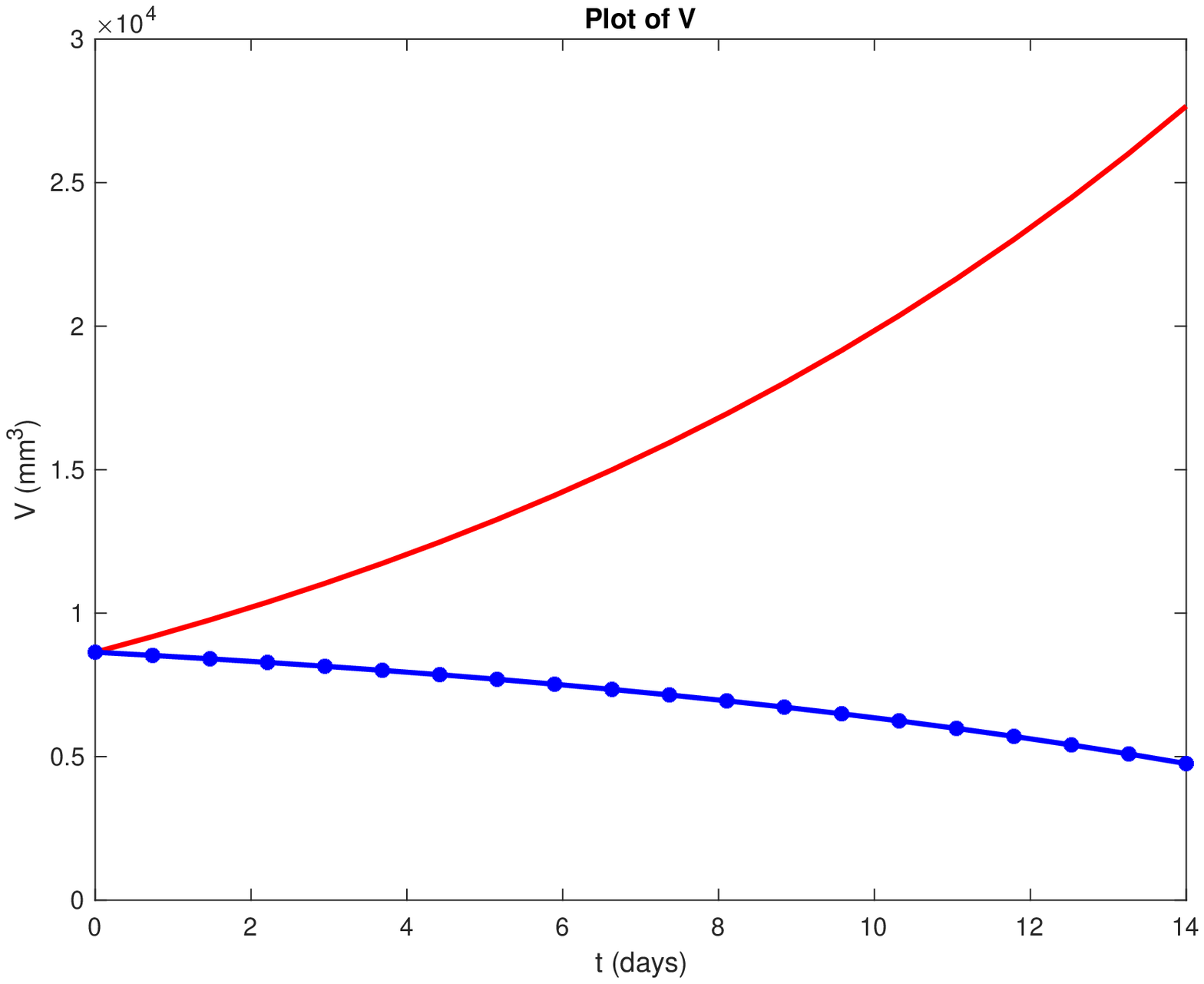}\label{V20_t2}}
\subfloat[$B$]{\includegraphics[width=0.35\textwidth, height=0.30\textwidth]{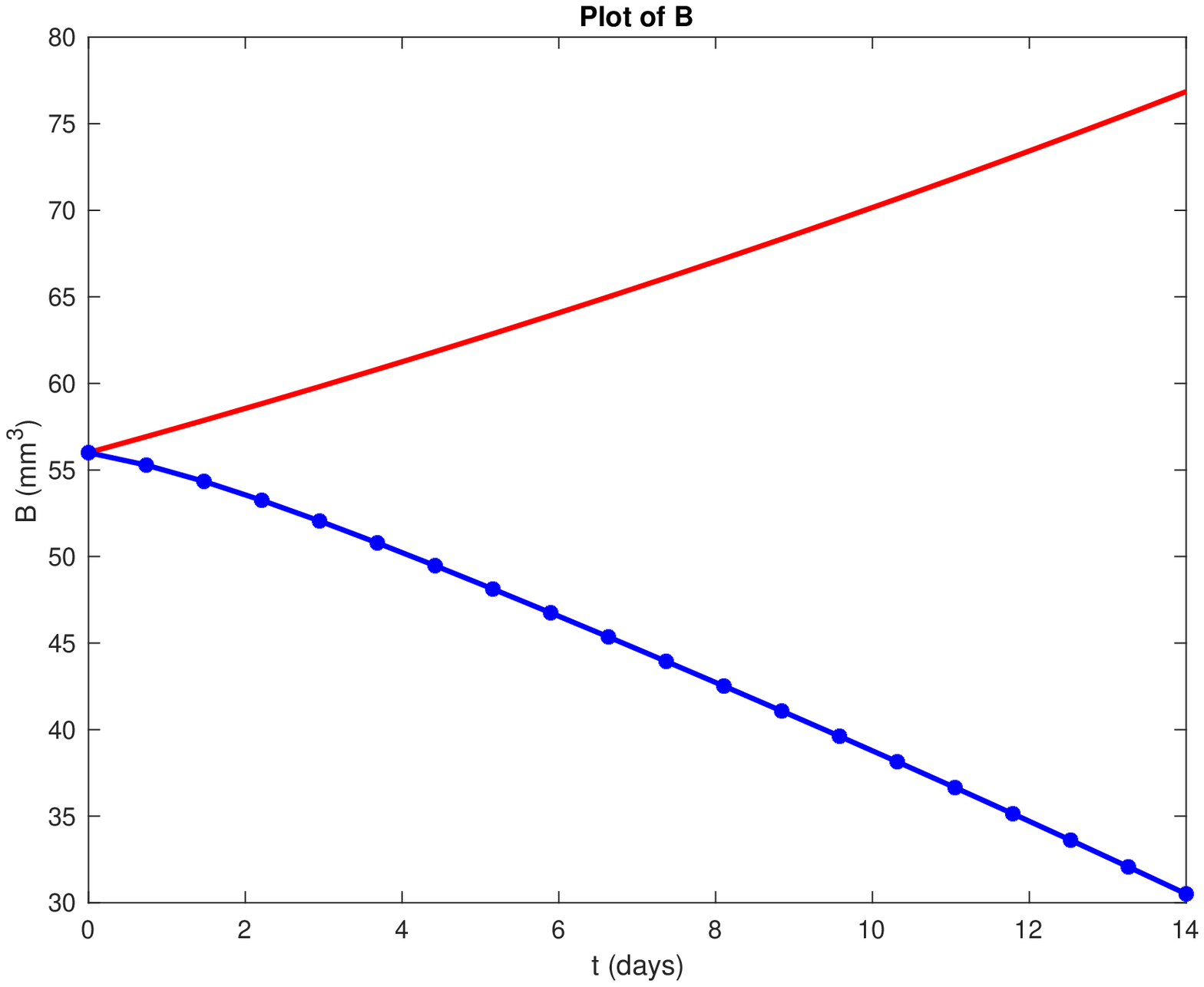}\label{B20_t2}}
\subfloat[$T$]{\includegraphics[width=0.35\textwidth, height=0.30\textwidth]{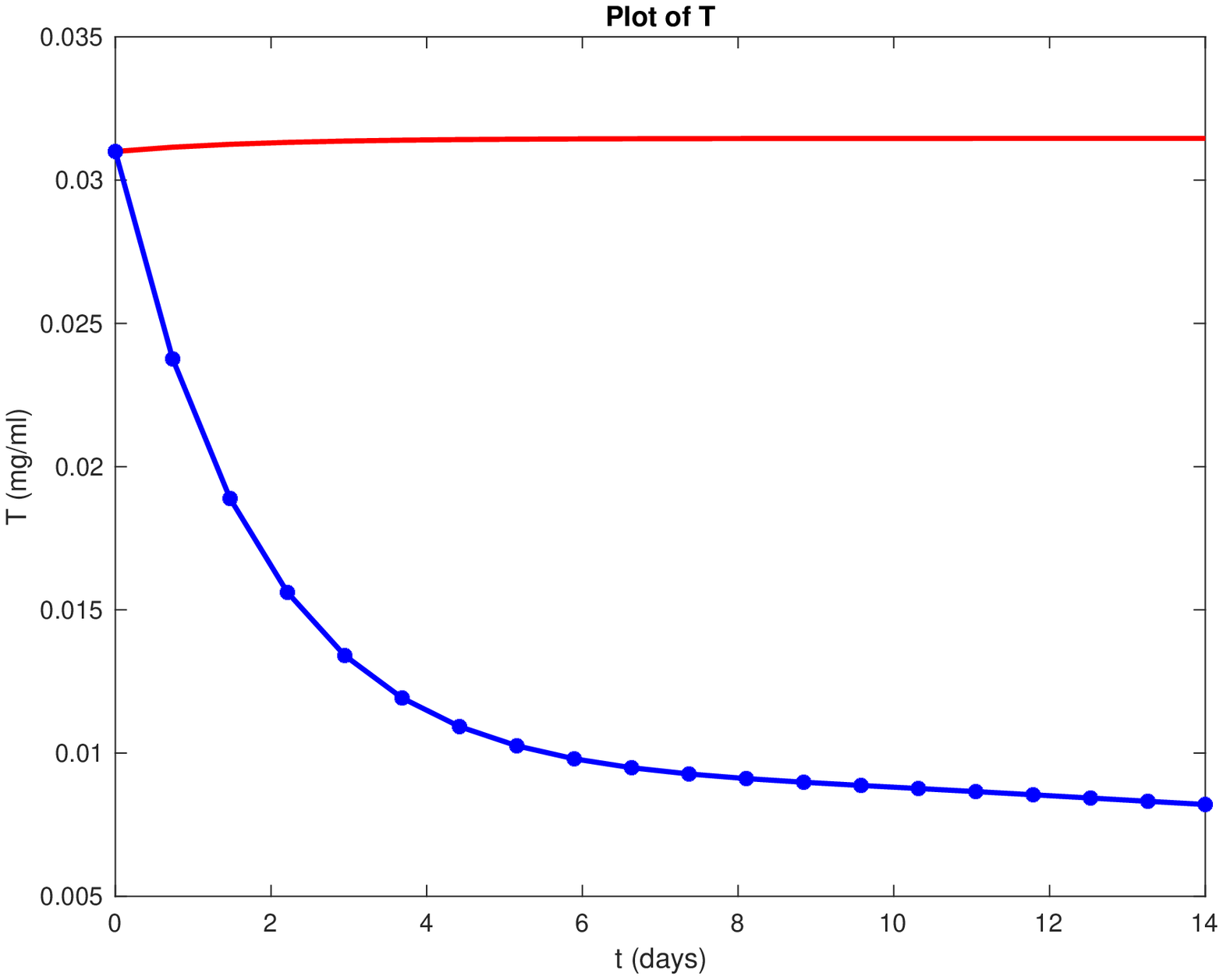}\label{T20_22}}
 \caption{Test Case 2: Mean trajectory plots of the profiles of $V,B,T$ with and without the combination therapy. Red curves indicate the profiles without treatment, blue curves indicate the profiles with treatment }
    \label{fig:test_case2_treat}
  \end{figure}

Figure \ref{fig:test_case2_treat} shows the plots of the mean PDFs with and without treatment. The red curves denotes profiles of $V,B,T$ without treatment, whereas the blue curves denotes profiles with the effect of treatments. We again note that in the absence of the combination therapy, the values of $V,B,T$ clearly rise uncontrolled, but on the administration of the combination therapy, $V,B,T$ is driven to the desired cancer-free state. Figure \ref{fig:test_case2_u} shows the dosage patterns of Capecitabine and Bevacizumab. As in Test Case 1, we note that initially, Capecitabine is administered with higher dosages but over time the dosage is significantly reduced. On the other hand, the dosage of Bevacizumab is low during the initial phases of the treatment and subsequently increases with the decrease of Capecitabine.

\begin{figure}[h]
\centering
\subfloat[$u1$: Capecitabine]{\includegraphics[width=0.35\textwidth, height=0.30\textwidth]{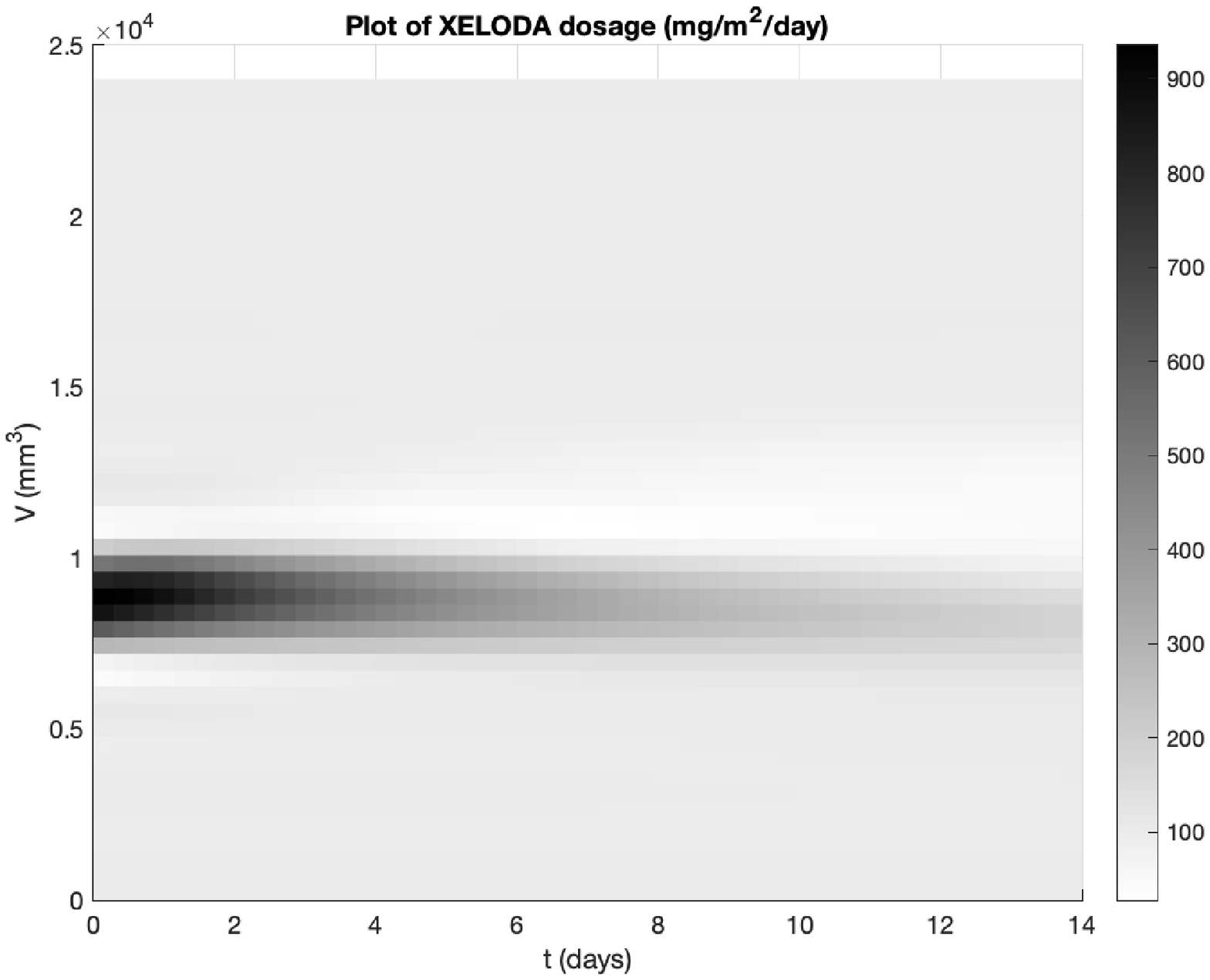}\label{u11}}
\subfloat[$u3$: Bevacizumab]{\includegraphics[width=0.35\textwidth, height=0.30\textwidth]{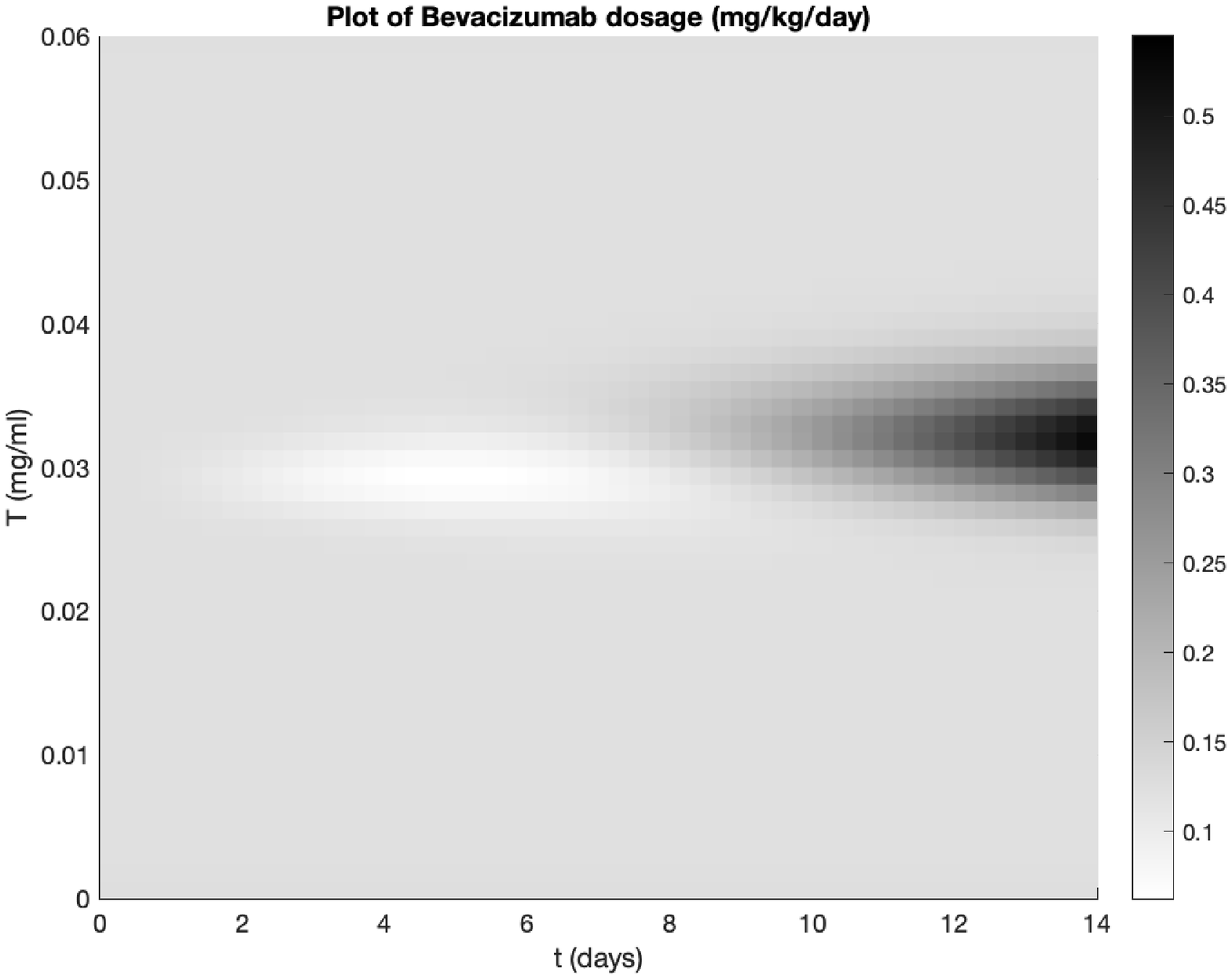}\label{u22}}
 \caption{Test Case 2: Feedback optimal combination treatment profiles for the real data case}
    \label{fig:test_case2_u}
  \end{figure}

\subsection{Discussion}
From the two aforementioned experiments, we note that a feedback combination strategy of Bevacizumab and Capecitabine provides an optimal treatment regime for colon cancer. This is consistent with previous experimental studies (e.g., see \cite{Ilic2016}), where it has been demonstrated that such combination therapies are potentially more effective than single therapies. However, it has also been noted (e.g., see \cite{Mohile2013}) that Bevacizumab, in combination with other chemotherapy drugs, might elevate secondary risk factors, like higher bleeding risk, proteinuria, hypertension, and arterial thromboses, especially for older people. Thus, it is quite important to administer a low dosage of both Bevacizumab and Capecitabine. From the numerical results, we see that for controlling colon cancer-induced angiogenesis, the maximum dosage of Capecitabine required was approximately 900 mg/m$^2$/day and of Bevacizumab was approximately 0.55 mg/kg/day. This is significantly lower than the current acceptable practice of administering 1250 mg/m$^2$/day dosage level of Capecitabine and 10 mg/kg repeated over a two week period. Thus, a feedback strategy and daily dosage administration results in a lower toxicity level and prevention of secondary risk factors. This observation is consistent with the results presented in \cite{Cser2019}. Finally, we observe from the numerically results that the concentration of Capecitabine decreases over time, and Bevacizumab takes over as the primary treatment mode till the end of the treatment cycle. This suggests that once colon cancer is controlled, it might be feasible to use Bevacizumab as the primary treatment in combination with a small chemotherapy dosage after the initial cycle (see https://www.avastin.com/patient/mcrc/treatment/length-of-time.html).

\section{Conclusion}
 
 In this paper, we presented a new framework for obtaining personalized optimal treatment strategies in colon cancer-induced angiogenesis. We considered the angiogenic pathway dynamics proposed in \cite{Cser2019} and extended it to a stochastic process to account for the random perturbations. We characterized the stochastic process using the PDF, whose evolution is governed by the FP equation. The coefficients in the FP equation represented the unknown patient specific parameters that we estimate using the patient data, by formulating a PDE-constrained optimization problem. The numerical discretization of the FP equations were done using a time-splitting scheme and Chang-Cooper spatial discretization method. We proved the properties of conservativeness, positivity and second order convergence of the numerical scheme. We also solved the optimality system using a projected NCG scheme. Furthermore, we studied the sensitivity analysis of the optimal parameters with respect to the tumor volume using the LHS-PRCC method. This in turn, helped us to incorporate appropriate combination therapies into the FP model. We solved an optimal control problem to obtain the optimal combination therapy. Numerical experiments, involving Bevacizumab and Capecitabine, with synthetic data and real data using experimental mice demonstrates that optimal combination therapies for cure of colon cancer-induced angiogenesis can be obtained real-time with high accuracy.

\section*{Acknowledgments} 

S. Roy and S. Pal express their thanks to National Cancer Institute of the National Institutes of Health (Award Number R21CA242933) for supporting this research. The research of Pan laboratory has been supported by National Institutes of Health Grant (Award Number R01 CA185055)

\end{document}